%% file: interval4.tex
\documentclass[12pt]{article}

\usepackage{import}

\usepackage{amsmath,amsfonts,amssymb,amscd,amsthm,latexsym}
\usepackage{graphicx}
\usepackage{color}
\usepackage[all]{xy}
\usepackage{fullpage}

\newcommand{\Al}{X}
\newcommand{\State}{V}

\newtheorem{theorem}{Theorem}
\newtheorem{proposition}{Proposition}
\newtheorem{lemma}{Lemma}
\newtheorem{corollary}{Corollary}

\theoremstyle{definition}
\newtheorem{example}{Example}

\newcommand{\refprop}[1]{Proposition~\ref{#1}}

\newcommand{\refthm}[1]{Theorem~\ref{#1}}
\newcommand{\reflem}[1]{Lem\-ma~\ref{#1}}

\newcommand{\refsec}[1]{Section~\ref{#1}}

\usepackage{mathbbol}
\newcommand{\unit}{\mathbb{1}}
\newcommand{\zero}{\mathbb{O}}
\newcommand{\Rectangle}[2]{#1 \times #2}

\usepackage[colorinlistoftodos]{todonotes}

\definecolor{bdcolor}{rgb}{0.8,0,0.2}
\definecolor{jdcolor}{rgb}{1,0,0}
\definecolor{ihcolor}{cmyk}{0.2,1,0,0}
\definecolor{gscolor}{cmyk}{0.2,1,0,0}

\begin{document}

\title{Convolution, Separation and Concurrency}

\author{Brijesh Dongol\\University of Sheffield\\United Kingdom \and
  Ian J. Hayes\\University of Queensland\\Australia \and Georg
  Struth\\University of Sheffield\\United Kingdom}

\renewcommand{\labelenumi}{(\alph{enumi})}

\maketitle

\begin{abstract}
  A notion of convolution is presented in the context of formal power
  series together with lifting constructions characterising algebras
  of such series, which usually are quantales. A number of examples
  underpin the universality of these constructions, the most prominent
  ones being separation logics, where convolution is separating
  conjunction in an assertion quantale; interval logics, where
  convolution is the chop operation; and stream interval functions,
  where convolution is used for analysing the trajectories of
  dynamical or real-time systems. A Hoare logic is constructed in a
  generic fashion on the power series quantale, which applies to each
  of these examples. In many cases, commutative notions of convolution
  have natural interpretations as concurrency operations.

\vspace{\baselineskip} 

  \end{abstract}

  \pagestyle{plain}


\section{Introduction}\label{sec:introduction}

Algebraic approaches play a fundamental role in mathematics and
computing.  Algebraic axioms for groups, rings, modules or lattices,
for instance, capture certain features of concrete models in an
abstract uniform fashion. Fundamental constructions, such as products,
quotients or adjunctions, can be presented and investigated in algebra
in simple generic ways.

This article investigates the notion of \emph{convolution} or
\emph{Cauchy product} from formal language
theory~\cite{Handbook,BerstelReutenauer} as such a fundamental
notion, supporting the generic construction of various models and
calculi that are interesting to computing. This provides a unified
structural view on various computational models known from the
computer science literature.

Questions of summability and divergence aside, the operational content
of convolution is simple: an entity is separated in all possible ways
into two parts, two functions are simultaneously applied to these
parts, their outputs are combined, and the sum over all possible
combinations is taken. Suppose two functions $f$ and $g$ from an
algebra $S$ (with suitable multiplication $\circ$) into an algebra $Q$
(with suitable multiplication $\odot$ and suitable summation ${\rm
  \Sigma}$). Using the nomenclature of formal language theory, the
convolution of $f$ and $g$ for an element $x\in S$ is  defined as
\begin{equation*}
  (f\otimes g)\, x \ = \sum_{x=y\circ z} f\, y \odot g\, z.
\end{equation*}
Hence $x$ is first separated in all possible ways into parts $y$ and
$z$. The function $f$ is then applied to $y$ and $g$ to $z$. After
that, the results of these applications are combined in $Q$. The
convolution is indeed the sum of all possible splittings of $x$.

In formal language theory, functions $f:S\to Q$ are also known as
power series---more precisely as formal or rational power series. This
notion is slightly different from that commonly used in algebra, as
are the notions of convolution or Cauchy product. In formal language theory,
moreover, power series usually map elements of the free monoid
$S=X^\ast$ over the finite alphabet $X$---the set of words or strings
over $X$---into a semiring $(Q,+,\odot,0)$. Since every word can only
be split into finitely many prefix/suffix pairs, the summation
occurring in convolution is finite and therefore well defined.  A
simple example of $Q$ is the boolean semiring with $+$ as disjunction
and $\odot$ as conjunction. Power series then become characteristic
functions representing languages, telling us whether or not some word
is in some language, and convolution becomes language product. In more
general settings, $Q$ can model probabilities or weights associated
to words; a Handbook has been devoted to the
subject~\cite{Handbook}. This example alone underpins the power of
power series and convolution.

Complementing this body of work, we generalise the typeof power
series, rebalancing the assumptions on source algebras $S$ and target
algebras $Q$ and thus shifting the focus to other applications.  Among
those, we show that, for suitable algebras $S$ and $Q$, convolution
becomes \emph{separating conjunction} of separation logic
(cf.~\cite{COY07}), or alternatively the \emph{chop} operator of
interval temporal logics~\cite{Mos00}. Both can in fact be combined,
for instance within interval logics, to provide new notions of
concurrency for this setting. In addition, we use power series to
capture, in a generic manner, the algebraic properties of convolution
for wide classes of instances and show how Hoare-style compositional
inference systems can be derived uniformly for all of them.

More concretely, the main contributions of this article are as
follows.
\begin{itemize}
\item Considering power series that map arbitrary partial semigroups
  into quantales, we prove a generic lifting result showing that
  spaces of power series form quantales as well.
\item This lifting result is generalised by making the target quantale
  partial, by considering bi-semigroups and bi-quantales with two
  multiplication operations, by mapping two separate semigroups into a
  bi-quantale, and by setting up source semigroups suitable for
  distinguishing between finite and infinite system behaviours.
\item We show that algebras of state and predicate transformers arise
  as instances of the generic lifting theorem.
\item Propositional Hoare calculi (without assignment axioms) are
  derived within the power series quantale in a generic fashion; and
  we discuss some ramifications of deriving concurrency rules in this
  setting.
\item We provide a series of instances of the lifting result, showing
  how quantales of languages, binary relations, matrices and automata,
  sets of paths and traces as well as interval functions and
  predicates arise from a non-commutative notion of convolution.
\item In the commutative case, we present the assertion quantales of
  separation logic with separation based on general resource monoids
  as well as multisets, sets with disjoint union and heaplets. We also
  present a separation operation on finite vectors, which leads to a
  notion of convolution-based parallelism for linear transformations.
\item Both kinds of instances are combined into a new algebraic
  approach to stream interval functions and predicates, which allow
  the logical analysis of trajectories of dynamic and real time
  systems. This provides a convolution-based spatial concurrency
  operation in addition to the conventional temporal chop operator.
\item We illustrate how convolution as separating conjunction allows
  us to derive the frame rule of separation logic by simple equational
  reasoning.
\end{itemize}

Our lifting results are generic in the following sense: after setting
up a suitable partial semigroup---words under concatenation, closed
intervals under chop, multisets under addition or resource monoids
under resource aggregation---the space of all functions into a
quantale automatically forms a quantale with convolution as
multiplication.  When the target quantale is formed by the booleans,
power series can be identified with and predicates and characteristic
functions for sets, as their extensions. Multiplication in the
booleans becomes conjunction and convolution then reduces to
\begin{equation*}
  (f\otimes g)\, x = \sum_{x=y\circ z} f\, y \sqcap g\, z.
\end{equation*}
If $S$ is a set of resources and $\circ$ a (commutative) notion of
resource aggregation, then convolution is separating conjunction. If
$S$ is a set of closed intervals and $\circ$ splits an interval into
two disjoint parts, then convolution is chop. In that sense,
separating conjunction can be seen as a language product over
resources and chop as a language product over intervals. Here and in
all similar cases, our lifting result implies that the predicates of
type $S \to \mathbb{B}$ form an assertion quantale; in the first case
that of separation logic; in the second one that of interval
logics. But our results cover models beyond the booleans, for instance
probabilistic or weighted predicates or other kinds of functions. In
general, the convolution has a strongly spatial and concurrent flavour
whenever the operations $\circ$ and $\odot$ are commutative.

Similarly, for all instances of this lifting, the construction of
Hoare logics is generic because it works for abitrary
quantales~\cite{HMSW11}. Finally, due to the emphasis on functions
instead of sets, the approach is constructive so long as the
underlying source and target algebras are.

The remainder of this article is organised as
follows. Section~\ref{sec:algebr-prel} recalls the basic algebraic
structures needed. Section~\ref{sec:fpsquantale} introduces our
approach to power series with partial semigroups as source algebras
and quantales as target algebras; it also proves our basic lifting
result. Section~\ref{sec:booleancase} discusses the case of power
series into the boolean quantale, when convolution becomes a possibly
non-commutative notion of separating
conjunction. Section~\ref{sec:fpsquantaleexamples}
and~\ref{sec:fpscomquantaleexamples} present non-commutative and
commutative instances of our lifting lemma;
Section~\ref{sec:fpsquantaleexamples} discussing, among others, the
chop operation over intervals and
Section~\ref{sec:fpscomquantaleexamples} focusing on variants of
separating conjunction. Section~\ref{sec:transformers} shows how state
and predicate transformers arise in the power series setting.
Section~\ref{sec:partial-formal-power} presents a lifting result for
power series into partial quantales with an
example. Section~\ref{sec:formal-power-series} generalises the lifting
result to bi-semigroups and bi-quantales and presents two
examples. Section~\ref{sec:fpsbiquantale} generalises the result to
power series from two semigroups into a bi-quantale;
Section~\ref{sec:biquantale-examples} presents in particular the
quantale of stream interval functions, which is based on this
generalisation.  Section~\ref{sec:futuristic} further generalises the
approach to applications with finite and infinite
behaviours. Section~\ref{sec:interchange} shows that the interchange
laws of concurrent Kleene algebras fail in general power series
quantales. Based on this, Section~\ref{sec:hoare} discusses how
generic Hoare logics can be developed over power series
quantales. Section~\ref{sec:frame} shows how the approach can be used
for deriving the frame rule of separation logic, using convolution as
the algebraic notion of separating
conjunction. Section~\ref{sec:conclusion} contains a conclusion.


\section{Algebraic Preliminaries}
\label{sec:algebr-prel}

In this section, we briefly recall the most important mathematical
structures used in this article: partial semigroups and monoids, their
commutative variants, semigroups and dioid as well as quantales. We
also consider such structures with two operations of composition or
multiplication, that is, bi-semigroups, bi-monoids, bi-semirings and
bi-quantales.

\paragraph{Semigroups.}
A \emph{partial semigroup} is a structure $(S,\cdot,\bot)$ such that
$(S,\cdot)$ is a semigroup and $x\cdot \bot =\bot = \bot \cdot x$
holds for all $x \in S$.  It follows that $\bot \notin S$, which is
significant for various definitions in this article. A \emph{partial
  monoid} is a partial semigroup with multiplicative unit $1$. We
often write $(S,\cdot)$ for partial semigroups and $(S,\cdot, 1)$ for
partial monoids, leaving $\bot$ implicit. A (partial) semigroup $S$ is
\emph{commutative} if $x\cdot y=y\cdot x$ for all $x,y\in
S$. Henceforth, we use $\cdot$ for a general multiplication and $\ast$
for a commutative one.

An important property of semigroups is \emph{opposition duality}.  For
every semigroup $(S,\cdot)$, the structure $(S,\odot)$ with $x\odot
y=y\cdot x$ for all $x,y\in S$ forms a semigroup; the \emph{opposite}
of $S$. Similarly, the opposite of a monoid is a monoid.

The definitions of semigroups and monoids generalise to $n$
operations, but we are mainly interested in the case $n=2$. A
\emph{partial bi-semigroup} is a structure $(S,\circ,\bullet)$ such
that $(S,\circ)$ and $(S,\bullet)$ are partial semigroups. 
\emph{Partial bi-monoids} $(S,\circ,\bullet,1,1')$ can be obtained from them as
standard.

\paragraph{Semirings.}
A \emph{semiring} is a structure $(S,+,\cdot,0)$ such that $(S,+,0)$
is a commutative monoid, $(S,\cdot)$ a semigroup, and the
distributivity laws $x\cdot (y+z)=x\cdot y+x\cdot z$ and $(x+y)\cdot
z= x\cdot z+y\cdot z$ as well as the annihilation laws $0\cdot x= 0$
and $x\cdot 0=0$ hold. A semiring is \emph{unital} if the
multiplicative reduct is a monoid (with unit $1$). A \emph{dioid} is
an additively idempotent semiring $S$, that is, $x+x=x$ holds for all
$x\in S$. The additive reduct of a dioid thus forms a semilattice with
order defined by $x\le y\Leftrightarrow x+y=y$.  Obviously, the
classes of semirings and dioids are closed under opposition duality.

A \emph{bi-semiring} is a structure 
$(S,+,\circ,\bullet,0)$ such that $(S,+,\circ,0)$ and 
$(S,+,\bullet,0)$ are semirings; a \emph{trioid} is an additively 
idempotent bi-semiring. A bi-semiring or trioid is \emph{unital} if 
the underlying bi-semigroup is a bi-monoid. 

\paragraph{Quantales.}
A \emph{quantale} is a structure $(Q,\le,\cdot)$ such that $(Q,\le)$
is a complete lattice, $(Q,\cdot)$ is a semigroup and the
distributivity axioms
\begin{equation*}
  x\cdot (\sum_{i\in I}y_i) = \sum_{i\in I} (x\cdot y_i),\qquad (\sum_{i\in I} x_i)\cdot y = \sum_{i\in I}(x_i\cdot y) 
\end{equation*}
hold, where $\sum X$ denotes the supremum of a set $X\subseteq
Q$. Similarly, we write $\prod X$ for the infimum of $X$.  The
distributivity laws imply, in particular, the  isotonicity
laws
\begin{equation*}
  x
  \le y  \Rightarrow   z\cdot x \le z\cdot y, \qquad
  x\le y  \Rightarrow  
  x\cdot z \le y\cdot z.
\end{equation*}  
A quantale is \emph{commutative} and
\emph{partial} if the underlying semigroup is as well; \emph{unital} if the
underlying semigroup is a monoid; and \emph{distributive} if the
infinite distributivity laws
\begin{equation*}
  x\sqcap (\sum_{i\in I} y_i) = \sum_{i\in I} (x\sqcap y_i),\qquad 
  x+ (\prod_{i\in I} y_i) = \prod_{i\in I} (x+ y_i) 
\end{equation*}
hold. A \emph{boolean quantale} is a distributive quantale in which
every element has a complement. 

The boolean unital quantale $\mathbb{B}$, where multiplication $\cdot$
coincides with meet, plays an important role in this article.

A \emph{bi-quantale} is a structure $(Q,\le,\circ,\bullet)$ such that
$(Q,\le,\circ)$ and $(Q,\le,\bullet)$ are quantales. It is unital if
the two underlying semigroups are monoids.

It is easy to see that every (unital) quantale is a (unital) dioid and
every (unital) bi-quantale a (unital) trioid. In particular,
$0=\sum\emptyset =\sum_{i\in\emptyset} x_i$ and annihilation laws as
in dioids follow from this as special cases of distributivity.


\section{Power Series Quantales}\label{sec:fpsquantale}

Formal (or rational) power series ~\cite{BerstelReutenauer} have been
studied in formal language theory for decades. For brevity, we call
them \emph{power series} in this article.
In formal language theory, a power series is simply a function from the
free monoid $\Al^\ast$ over a finite alphabet $\Al$ into a suitable
algebra $Q$, usually a semiring or dioid $(Q,+,\cdot,0,1)$.

Operations on $f,g:\Al^\ast\to Q$ are defined as follows. Addition is
lifted pointwise, that is, $(f+g)\, x = f\, x + g\, x$. Multiplication is
given by the \emph{convolution} or \emph{Cauchy product}
\begin{equation*}
  (f\cdot g)\, x =\sum_{x=yz} f\, y\cdot g\, z,
\end{equation*}
where $yz$ denotes word concatenation and the sum in the convolution
is finite since finite words can only be split in finitely many ways
into prefix/suffix pairs. Furthermore, the \emph{empty power series}
$\zero$ maps every word to $0$, whereas the \emph{unit power series}
$\unit$ maps the empty word to $1$ and all other words to $0$.

We write $Q^{\Al^\ast}$ for the set of power series from $\Al^\ast$ to
$Q$ and, more generally, $Q^S$ for the class of functions of type
$S\to Q$. The following lifting result is well known.
\begin{proposition}\label{prop:fpslifting}
  If $(Q,+,\cdot,0,1)$ is a semiring (dioid), then so is
  $(Q^{\Al^\ast},+,\cdot,\zero,\unit)$.
\end{proposition}

This construction generalises from free monoids over finite alphabets
to arbitrary partial semigroups or monoids.  The sums in convolutions
then become infinite due to infinitely many possible decompositions of
elements. Here, due to potential divergence, these sums may not
exist. However, we usually consider target algebras in which addition
is idempotent and sums corresponds to suprema. The existence of
arbitrary suprema can then be covered by completeness assumptions.

We fix suitable algebraic structures $S$ and $Q$.  First, we merely
assume that $S$ is a set, but for more powerful lifting results it is
required to be a partial semigroup or partial monoid.

For a family of functions $f_i:S\to Q$
and $i\in I$ we define
\begin{equation*}
  (\sum_{i\in I} f_i)\, x = \sum_{i\in I} f_i\, x,
\end{equation*}
whenever the supremum in $Q$ at the right-hand side exists. This
comprises 
\begin{equation*}
(f+g)\, x = f\, x + g\, x
\end{equation*}
as a special case.  Since $x$ ranges over $S$, the constant $\bot$ is
excluded as a value.  Another special case is
\begin{equation*}
  (\sum_{i\in\emptyset} f_i)\, x = (\sum \emptyset)\, x =
  \sum_{i\in\emptyset} f_i\, x = 0.
\end{equation*}
Hence, in particular, $\sum_{i\in\emptyset} f_i=\lambda x.\ 0$ and we
write $\zero$ for this function.

We define the convolution
\begin{equation*}
  (f\cdot g)\, x = \sum_{x=y\cdot z} f\, y\cdot g\, z,
\end{equation*}
where the multiplication symbol is overloaded to be used on $S$, $Q$
and $Q^S$. Again, this requires that the supremum in the right-hand
side exists in $Q$.  In the expression $x=y\cdot z$, the constant
$\bot$ is again excluded as a value. Undefined splittings of $x$ are
thus excluded from contributing to convolutions.

Finally, whenever $S$ and $Q$ are endowed with suitable units, we
define $\unit : S\to Q$ as
\begin{equation*}
  \unit\, x = 
  \begin{cases}
    1, & \text{if } x = 1,\\
    0, & \text{otherwise},
  \end{cases}
\end{equation*}
as for formal languages. 

Theorem~\ref{thm:quantale-lifting}, the main result in this section,
shows that quantale laws lift from the algebra $Q$ to the function
space $Q^S$ of power series under these definitions. On the way to
this result we recall that semilattice and lattice structures lift to
function spaces, a fundamental result of domain
theory~\cite{AbramskyJung}.

\begin{lemma}\label{lem:semilattice-lifting}
  Let $S$ be a set. If $(L,+,0)$ is a semilattice with least element
  $0$ then so is $(L^S,+,\zero)$. If $L$ is a complete lattice, then so is
  $L^S$.
\end{lemma}
\begin{proof}~ The semilattice lifting is covered by
  Proposition~\ref{prop:fpslifting}. As usual, $L^S$ is ordered by
  $f\le g\Leftrightarrow f+g=g$, and $\zero \le f$ for all $f\in L^S$.

  If arbitrary suprema exist in $L$, then completeness lifts to
  $L^S$ by definition of $\sum_{i\in I}f_i$.  Finally, every complete
  join-semilattice is a complete lattice.
\end{proof}

Infima, if they exist, are defined like suprema by pointwise lifting as
\begin{equation*}
  (\prod_{i\in I}f_i)\, x = \prod_{i\in I}f_i\, x,
\end{equation*}
thus $(f\sqcap g)\, x = f\l x\sqcap g\
x$. \reflem{lem:semilattice-lifting} can then be strengthened.
\begin{lemma}\label{lem:lattice-lifting}
  Let $S$ be a set. If $(D,+,\sqcap,0)$ is a (distributive) lattice with
  least element $0$, then so is $(D^S,+,\sqcap,\zero)$. Completeness and infinite distributivity laws between infima and suprema lift from $D$ to $D^S$.
\end{lemma}
\begin{proof}
  The join- and meet-semilattice laws for $+$ and $\sqcap$ follow from
  Lemma~\ref{lem:semilattice-lifting}. We need to verify absorption
  and distributivity. Let $f,g,h:S\to D$ and $x \in S$.
  \begin{itemize}
  \item $(f\sqcap (f+g))\, x = f\, x \sqcap (f\, x + g\, x) = f\, x$ by
    absorption on $D$. The proof of $f+(f\sqcap
    g)=f$ is lattice dual.
  \item The finite distributivity laws are special cases of the
    infinite ones below.
  \end{itemize}
  Completeness is covered by \reflem{lem:semilattice-lifting}. For
  infinite distributivity,
\begin{equation*}
  (f\sqcap \sum_{i\in I} g_i)\, x = f\, x \sqcap \sum_{i\in I}
  g_i\, x = \sum_{i\in I} f\, x \sqcap g_i\, x = \sum_{i\in I} (f\sqcap
  g_i)\, x = (\sum_{i\in I} f\sqcap g_i)\, x.
\end{equation*} 
The other distributivity law then follows from lattice duality.
\end{proof}

The final lifting result in this section deals with multiplicative
structure as well. This requires $S$ to be a partial semigroup instead
of a set.
\begin{theorem}\label{thm:quantale-lifting} Let $(S,\cdot)$ be a 
  partial semigroup. If $(Q,\le,\cdot)$ is a (distributive) quantale,
  then so is $(Q^S,\le,\cdot)$. In addition, commutativity in $Q$
  lifts to $Q^S$ if $S$ is commutative; unitality in $Q$ lifts to 
  $Q^S$ if $S$ is a partial monoid. 
\end{theorem}
\begin{proof}
  Since $Q$ is a quantale, all infinite suprema and infima exist; in
  particular those needed for convolutions.

  The lifting to complete (distributive) lattices is covered by
  Lemma~\ref{lem:lattice-lifting}.  It therefore remains to check the
  multiplicative monoid laws, distributivity of multiplication and
  annihilation. For left distributivity, for instance,
    \begin{align*}
      (f\cdot \sum_{i\in I} g_i)\, x 
      = \sum_{x=y\cdot z} f\, y\cdot \sum_{i\in I} g_i\, z
      = \sum_{\substack{x=y\cdot z,\\ i\in I}} f\, y\cdot g_i\, z
      = \sum_{i\in I} (f\cdot g_i)\, x.
    \end{align*}
    The proof of right distributivity is opposition dual. 

    Left distributivity ensures associativity, the proof of which
    lifts as with rational power series
    (Proposition~\ref{prop:fpslifting}).  The restriction to partial
    semigroups is insignificant as, in $x= y\cdot z$, the constraint
    $x\in S$ only rules out contributions of $y\cdot z=\bot$. The same
    holds for unitality proofs.

    Commutativity lifts from $S$ and $Q$ as follows:
\begin{equation*}
 (f\cdot g)\, x = \sum_{x=y\cdot z} f\, y \cdot g\, z
  = \sum_{x=z\cdot y} g\, z \cdot f\, y =(g\cdot f)\, x.
\end{equation*}
\end{proof}

Once more the distributivity laws on $Q^S$ imply the annihilation laws
$\zero \cdot f= \zero$ and $f\cdot \zero = \zero$ for all $f:S\to Q$. When only finite
sums are needed, $Q$ can be assumed to be a semiring or dioid instead
of a quantale.  The following corollary to
\refthm{thm:quantale-lifting} provides an example.
\begin{corollary}\label{cor:quantale-lifting-finite} Let $(S,\cdot)$ be a finite
  partial semigroup. If $(Q,+,\cdot,0)$ is a semiring, then so is
  $(Q^S,+,\cdot,\zero)$. In addition, idempotency in $Q$ lifts to $Q^S$;
  commutativity in $Q$ lifts to $Q^S$ if $S$ is commutative; unitality
  in $Q$ lifts to $Q^S$ if $S$ is a partial monoid.
\end{corollary}
As another specialisation, \refprop{prop:fpslifting} is recovered
easily when $S$ is the free monoid over a given alphabet and $Q$ a
semiring or dioid.

\section{Power Series into the  Boolean Quantale}\label{sec:booleancase}

In many applications, the target quantale $Q$ is formed by the
booleans $\mathbb{B}$. Power series are then of type $S\to\mathbb{B}$
and can  be interpreted as characteristic functions or
predicates. In fact, $\mathbb{B}^S$ is isomorphic to the power set of
$S$, which, in turn is in one to one correspondence with the set of
all predicates over $S$, identifying predicates with their extensions.

In this context, \refthm{thm:quantale-lifting} specialises to the
powerset lifting of a partial semigroup or monoid $S$. For each $x\in
S$, the boolean value $f\, x$ expresses whether or not $x$ is in the set
corresponding to $f$. Powerset liftings have been studied widely in
mathematics~\cite{Goldblatt,Brink}. They have various applications in
program semantics, for instance as power domains (cf.~\cite{AbramskyJung}).

\begin{corollary}\label{cor:powerset-lifting}
  Let $S$ be a partial (commutative) semigroup. Then $\mathbb{B}^S$
  forms a (commutative) distributive quantale where $\mathbb{B}^S\cong
  2^S$, $\le$ corresponds to $\subseteq$ and convolution $\cdot$ to
  the complex product
  \begin{equation*}
    X\cdot Y =\{x\cdot y \mid x\in X\wedge y\in Y\}
  \end{equation*}
  for all $X,Y\subseteq S$. If $S$ has unit $1$, then $\mathbb{B}^S$
  has unit $\{1\}$.
\end{corollary}
Various instances of Corollary~\ref{cor:powerset-lifting} are discussed
in Sections~\ref{sec:fpsquantaleexamples}
and~\ref{sec:fpscomquantaleexamples}.

The quantale $\mathbb{B}^S$ carries a natural logical structure with
elements of $\mathbb{B}^S$ corresponding to predicates, suprema to
existential quantification, infima to universal quantification and the
lattice order to implication.  In particular, $+$ corresponds to
disjunction and $\sqcap$ to conjunction.

More interesting is the logical interpretation of convolution
\begin{equation*}
(f\cdot g)\, x =\sum_{x=y\cdot
  z} f\, y\cdot g\, z
\end{equation*}
in the boolean quantale $\mathbb{B}^S$. The expression $x= y\cdot z$
denotes the decomposition or separation of the semigroup element $x$
into parts $y$ and $z$. The composition $f\ y\cdot g\ z=f\ y\sqcap g\
z$ in $\mathbb{B}$ models the conjunction of predicate $f$ applied to
$y$ with predicate $g$ applied to $z$. Finally, the supremum $\sum$
models the existential quantification over these
conjunctions with respect to all possible decompositions of $x$.

The commutative case of Corollary~\ref{cor:powerset-lifting} is
immediately relevant to separation logic. In this context, the partial
commutative semigroup $(S,\ast)$ is know as the \emph{resource
  semigroup}~\cite{COY07}; it provides an algebraic abstraction of the
heap. Its powerset lifting $\mathbb{B}^S$ captures the algebra of
resource predicates that form the assertions of an extended Hoare
logic---the assertion quantale of separation logic.  In this assertion
quantale, separating conjunction is precisely convolution: the product
$x=y\ast z$ on the resource semigroup $S$ decomposes or separates the
resource or heap $x$ into parts of heaplets $y$ and $z$ and the
product $f\ y\ast g\ z=f\ y\sqcap g\ z$ in $\mathbb{B}$ once more
conjoins $f\ y$ and $g\ z$; hence $x=y\ast z$ separates whereas $f\
y\ast g\ z=f\ y\sqcap g\ z$ conjoins.  The concrete case of the heap
is considered in more detail in
Example~\ref{ex:separating-conjunction-heaplets}.

The power series approach thus yields a simple algebraic view on a
lifting to function spaces in which the algebraic operation of
convolution into the booleans allows various interpretations,
including that of a complex product, that of separating
conjunction---commutative or non-commutative---and that of separating
conjunction as a complex product. In the commutative setting it gives
a simple account of the category-theoretical approach to O'Hearn and
Pym's logic of bunched implication~\cite{OHearnP99} in which
convolution corresponds to coends and the quantale lifting is embodied
by Day's construction~\cite{Day}.


\section{Non-Commutative Examples}\label{sec:fpsquantaleexamples}

After the conceptual development of the previous sections we now
discuss a series of examples which underpin the universality and
relevance of the notion of convolution in computing. All of them can
be obtained as instances of \refthm{thm:quantale-lifting} after
setting up partial semigroups or monoids appropriately. For all these
structures, the lifting to the function space is then generic and
automatic. The booleans often form a particularly interesting target
quantale.

This section considers only examples with a non-commutative notion of
convolution; for commutative examples see Section~\ref{sec:fpscomquantaleexamples}.

\begin{example}[Formal Languages]\label{ex:formal-languages}
  Let $(\Al^\ast,\cdot,\varepsilon)$ be the free monoid generated by
  the finite alphabet $\Al$ with $\varepsilon$ denoting the empty
  word. Let $Q$ form a distributive unital quantale. Then
  $Q^{\Al^\ast}$ forms a distributive unital quantale as well by
  Theorem~\ref{thm:quantale-lifting}. More precisely, since suprema in
  convolutions are always finite, one obtains the unital dioid
  $(Q^{\Al^\ast},+,\cdot,\zero,\unit)$ by lifting from a dioid
  $(Q,+,\cdot,0,1)$. This is the well known rational power series
  dioid of formal language theory. For $Q=\mathbb{B}$ one obtains, by
  Corollary~\ref{cor:powerset-lifting}, the quantale
  $\mathbb{B}^{\Al^\ast}$ of formal languages over $\Al$. \qed
\end{example}

\begin{example}[Binary Relations]\label{ex:binary-relations}
  For a set $A$ consider the partial semigroup $(A\times A,\cdot)$
  with $\cdot$ defined, for all $a,b,c,d\in A$, by
  \begin{equation*}
    (a,b)\cdot
  (c,d)=
  \begin{cases}
    (a,d), & \text{ if } b=c,\\
    \bot, & \text{ otherwise}.
  \end{cases}
\end{equation*}
For $Q=\mathbb{B}$, Theorem~\ref{thm:quantale-lifting} (or its
Corollary~\ref{cor:powerset-lifting}) ensures that
$(\mathbb{B}^{A\times A},\le,\cdot)$, which is isomorphic to
$(2^{A\times A},\subseteq,\cdot)$, is the quantale of binary relations
under union, intersection, relational composition and the empty
relation.

More specifically, with every power series $f$ we associate a binary
relation $R_f$ defined by $(a,b)\in R_f \Leftrightarrow f\
(a,b)=1$. The empty relation $\emptyset$ obviously corresponds to the
power series defined by $\zero\, (a,b) = 0$ for all $a,b\in A$. Relational
composition is given by convolution
\begin{equation*}
  (f\cdot g)\, (a,b)=\sum_{c\in A} f\, (a,c)\cdot g\, (c,b).
\end{equation*}
It can then be checked  that $R_{f\cdot g}=R_f\cdot R_g=\{(a,b)
\mid \exists c. (a,c) \in R_f\wedge (c,b)\in R_g\}$.

The unit relation cannot be lifted from a unit in $A\times A$ because
$A\times A$ has no unit. Instead it can be defined on
$\mathbb{B}^{A\times A}$ directly as
\begin{equation*}
  \unit\, (a,b) = 
  \begin{cases}
    1, &\text{ if } a=b,\\
    0, &\text{ otherwise}.
  \end{cases}
\end{equation*}
\qed
\end{example}

The constructions for relations generalise, for instance, to
probabilistic or fuzzy relations where $Q\neq\mathbb{B}$, but this is
not explored any further. Instead we consider the case of matrices.

\begin{example}[Matrices]\label{ex:matrices}
  Matrices are functions $f:A_1\times A_2\to B$, where $A_1$ and $A_2$
  are index sets and $Q$ is a suitable coefficient algebra. For the
  sake of simplicity we restrict our attention to square matrices with
  $A_1=A_2= A$. General non-square matrices require more complex
  partiality conditions.

  The development is similar to binary relations, but uses coefficient
  algebras beyond $\mathbb{B}$. It is easy to check that matrix
  addition is modelled by
\begin{equation*}
(f+g)\, (i,j)= f\, (i,j)+g\, (i,j),
\end{equation*}
whereas matrix multiplication is given by
convolution
\begin{equation*}
  (f\cdot g)\, (i,j)=\sum_{k\in A}f\, (i,k)\cdot g\, (k,j),
\end{equation*}
under suitable restrictions to guarantee the existence of sums, such
as finiteness of $A$ or idempotency of additionin $Q$. The zero and
unit matrices are defined as in the relational case.
\begin{equation*}
  \unit\, (i,j) =
  \begin{cases}
    1,& \text{ if } i=j,\\
    0, & \text{ otherwise},
  \end{cases}
\qquad\qquad
\zero\, (i,j) = 0.
\end{equation*}
\refthm{thm:quantale-lifting} then shows that quantales are closed
under matrix formation. It can easily be adapted to showing that square
matrices of finite dimension over a semiring form a semiring or that
matrices over a dioid form a dioid.\qed
\end{example}

This example not only links matrices with  power series, it also
yields a simple explanation of the well known relationship between
binary relations and boolean matrices. If a relation $R\subseteq
A\times A$ is modelled as $f_R:A\times A\to \mathbb{B}$ defined by
$f_R\ (a,b)=1 \Leftrightarrow (a,b)\in R$ as indicated above, then it
\emph{is} a boolean matrix.

\begin{example}[Finite Automata]\label{ex:finite-automata} 
  Suppose $\State$ is a set of state symbols, $\Al$ an alphabet, $i
  \in \State$ the initial state and $F \subseteq \State$ a set of
  final states. Conway~\cite{Conway71} has shown that transition
  relations $\delta$ of finite automata $(\State,\Al,\delta,i,F)$ can
  be modelled in terms of finite matrices of type $\State\times
  \State\to \mathsf{Rex}(\Al)$ into the algebra of regular expressions
  $\mathsf{Rex}(\Al)$ over $\Al$, for instance a Kleene algebra with
  constants from $\Al$. Consider the following automaton and
  transition matrix as an example.
\begin{equation*}
\def\labelstyle{\normalsize}
\xymatrix{
 {}\ar[r] &*++[o][F-]{1} \ar^{a,b}@(ul,ur)\ar^b[r] & *++[o][F-]{2} \ar^a[r] & *++[o][F=]{3}
}
\qquad\qquad
   \begin{pmatrix}
     a+b&b&0\\
     0&0&a\\
     0&0&0
   \end{pmatrix}
\end{equation*}
More generally, the full automaton, including its initial and final
state information, is captured by the following triple.
\begin{equation*}
\left[
\begin{pmatrix}
1\\
0\\
0
\end{pmatrix},
   \begin{pmatrix}
     a+b&b&0\\
     0&0&a\\
     0&0&0
   \end{pmatrix},
\begin{pmatrix}
0\\
0\\
1
\end{pmatrix}
\right]
\end{equation*}

It is well known that the algebra of regular expressions forms a
dioid, hence Theorem~\ref{thm:quantale-lifting} applies, showing that
transition matrices over the dioid of regular expressions form a
dioid, as in Example~\ref{ex:matrices}. Other kinds of automata, such
as probabilistic or weighted ones, can be modelled along this
line.\qed
\end{example}
In fact, it has been shown that Kleene algebras are closed under
matrix formation~\cite{Kozen91}, but the neccessary treatment of the
Kleene star is beyond the scope of this article.  In addition, it is
well known that regular languages need not be closed under general
unions, hence do not form quantales.

\begin{example}[Trace Functions]\label{ex:traces}
  Let $\State$ be a finite set of state symbols and $\Al$ a finite set
  of transition symbols, as in a finite automaton. A
  \emph{trace}~\cite{Eilenberg} is a finite word over
  $(\State\cup\Al)^\ast$ in which state and transition symbols
  alternate, starting and ending with state symbols. We write
  $T(\State,\Al)$ for the set of traces over $\State$ and $\Al$. It is
  endowed with a partial monoid structure by defining, for
  $p_1\alpha_1q_1,p_2\alpha_2q_2\in T(\State,\Al)$, the \emph{fusion
    product}
  \begin{equation*}
    p_1\alpha_1q_1\cdot p_2\alpha_2q_2 =
    \begin{cases}
      p_1\alpha_1q_1\alpha_2q_2, & \text{ if } q_1=p_2,\\
      \bot, & \text{ otherwise}.
    \end{cases}
  \end{equation*}
  Then convolution becomes 
\begin{equation*}
  (f\cdot g)\, \tau = \sum_{\tau=p\alpha_1 r\cdot r\alpha_2 q} f\, p\alpha_1 r\cdot g\, r\alpha_2 q
\end{equation*}
and Theorem~\ref{thm:quantale-lifting} implies that the set
$Q^{T(\State,\Al)}$ of trace functions into the distributive quantale
$Q$ forms a distributive quantale. If $Q$ is unital, then
$Q^{T(\State,\Al)}$ becomes unital by defining
\begin{equation*}
  \unit\, x =
  \begin{cases}
    1, & \text{ if } x\in \State,\\
    0, & \text{ otherwise}.
  \end{cases}
\end{equation*}
For $Q=\mathbb{B}$ we obtain the well known quantale of sets of
traces.

Trace functions $\mathbb{B}^{T(\Al,\State)}$ have a natural
interpretation as trace predicates. Convolution $(f\cdot g)\, \tau$
indicates the various ways in which property $f$ holds on a prefix of
trace $\tau$ whereas property $g$ holds conjunctively on the consecutive
suffix, as for instance in temporal logics over computation traces or
paths.\qed
\end{example}

Sets of traces generalise both languages and binary relations, which
are obtained by forgetting structure in the underlying partial
monoid. Another special case is given by sets of paths in a graph,
which is obtained by forgetting state labels. The explicit
construction of the corresponding paths quantale is straightforward
and therefore not shown.

\begin{example}[Interval Functions]\label{ex:interval-functions}
  Let $(P,\le)$ be a linear order and $I_P$ the set of all closed
  intervals over $P$---the empty interval being open by definition.
  For an interval $x$, let $x_{min}$ and $x_{max}$ represent 
  respectively the minimum and maximum value in $x$.

  We impose a partial semigroup structure on $I_P$ be defining the
  \emph{fusion product} on $I_P$, similar to the case of binary
  relations, traces and matrices, as
\begin{equation*}
  x \cdot y = 
    \begin{cases}
      x \cup y, & \text{if } x_{max} = y_{min}, \\
      \bot,       & \text{otherwise}.
    \end{cases}
\end{equation*}
An \emph{interval function} is a function $f:I_P\to Q$ into a suitable
algebra. Whenever $Q$ is a (distributive) quantale,
Theorem~\ref{thm:quantale-lifting} applies and $Q^{I_P}$ forms a
(distributive) quantale, too. Convolution of interval functions is
given by
\begin{equation*}
  (f\cdot g)\, x = \sum_{x = y \cdot z} f\, y \cdot g\, z. 
\end{equation*}
Like in the case of relations, the unit interval function is not
lifted from $I_P$, but defined directly as
\begin{equation*}
  \unit\, [a,b] = 
  \begin{cases}
    1, & \text{ if } a=b,\\
    0, & \text{ otherwise}.
  \end{cases}
\end{equation*}
The quantale of interval functions then becomes unital.

\emph{Interval predicates} are functions of type $I_P\to
\mathbb{B}$. Convolution of interval predicates is known as the
\emph{chop} operation~\cite{Mos00}, where $(f \cdot g)\, [a, c]$ holds if
it is possible to split interval $[a, c]$ into $[a, b]$ and $[b, c]$
such that $f\, [a, b]$ and $g\, [b, c]$ hold in conjunction.
\begin{center}
  \scalebox{1}{\input{chop.pspdftex}}
\end{center}

The meaning of an interval predicate $f\, x$ can be defined in
various ways. For instance $f$ can hold somewhere (at some point) in
$x$ or (almost) everywhere (see \cite{Mos00,ZH04}), and it is even
possible to define and use non-deterministic evaluators \cite{HBDJ13}
that enable calculations of apparent states (see \cite{DHD14}). \qed
\end{example}

Naive use of interval predicates may have undesired effects: If $f\,
x$ means that $f$ holds at each point in interval $x$, then
$(f\cdot\neg f)$ is always false, since both $f$ and $\neg f$ would
have to hold in at least one fusion point, which is impossible. An
alternative definition of interval composition without fusion
therefore seems desirable.

The duration calculus presents a solution in terms of an `almost
everywhere' operator, such that a property holds almost everywhere in
an interval if it is false in the interval for a set of points of
measure zero \cite{ZH04}. Others have defined `jump conditions'
leaving the possibility of both $f$ and $\neg f$ holding at the fusion
point open \cite{HM09}.  Here we model a third approach \cite{DHD14},
with chop formalised over non-overlapping intervals, in the power
series setting.

\begin{example}[Intervals without Fusion]\label{ex:intervals-no-fusion}
  We define a composition of contiguous intervals that avoids
  fusion. To this end we consider the set $I_P$ of intervals of the
  form $(a,b)$, $(a,b]$, $[a,b)$ and $[a,b]$, for $a,b\in P$.
  We include the empty interval
  $\emptyset$, which is by definition equal to $(a,a)$, $(a,a]$ and
  $[a,a)$ for all $a\in P$.
  The interval $x$ precedes the interval $y$, written $x \prec y$,
  if $\forall a \in x, b \in y.\ a < b$.
  The composition of intervals is defined as
  \begin{equation*}
    x \cdot y =
      \begin{cases}
        x \cup y, & \text{if } x \cup y \in I_P \text{ and } x \prec y, \\
        \bot,       & \text{otherwise}.
      \end{cases}
  \end{equation*}

  Convolution $f\cdot g$ is then defined as
  usual. Theorem~\ref{thm:quantale-lifting} ensures once more that
  $Q^{I_P}$ forms a distributive quantale whenever $Q$ does.  The unit
  $\unit :I_P\to Q$, however, requires modification. Defining
  \begin{equation*}
    \unit\, x = 
    \begin{cases}
      1, &\text{ if } x=\emptyset,\\
      0, &\text{ otherwise},
    \end{cases}
  \end{equation*}
  it is easy to check that $(\unit \cdot f)\ x = f\, x = (f\cdot
  \unit)\, x$ for any interval $x$ and the new definition of interval
  composition. This makes the quantale $Q^{I_P}$ unital. \qed
\end{example}

The examples in this section show that the generic lifting
construction in Theorem~\ref{thm:quantale-lifting} allows a uniform
treatment of a variety of mathematical objects, including relations,
formal languages, matrices and sets of intervals. In each case, a
(partial) composition on the underlying objects needs to be defined,
e.g., on words, ordered pairs, index pairs of matrices, traces, paths
or intervals. Lifting to the function space is then generic.

Such a generic lifting has been discussed previously for languages,
relations, paths and traces in the context of an Isabelle/HOL library
with models of Kleene
algebras~\cite{ArmstrongSW-JLAMP,ArmstrongSW13}. Theorem~\ref{thm:quantale-lifting}
has, in fact, already been implemented in Isabelle.  Based on this,
the existing implementation of models of Kleene algebras can be
unified and simplified considerably.



\section{Commutative Examples}\label{sec:fpscomquantaleexamples}

This section provides instances of Theorem~\ref{thm:quantale-lifting}
and Corollary~\ref{cor:powerset-lifting} for the commutative case. As
discussed in Section~\ref{sec:booleancase},  this situation typically
arises when the composition of the underlying semigroup $(S,\ast)$ is
used to split resources, heaps, states, etc, in a spatial fashion,
which is in contrast to the previous section where $f \cdot g$ meant
that there was a dependency between $f$ and $g$, which often carries a
temporal meaning. One can often think of convolution instantiated to
such a spatial separation in terms of parallelism or concurrency.

In particular we instantiate Theorem~\ref{thm:quantale-lifting} to
four kinds of resource monoids based on multisets under multiset
union, sets under disjoint union, partial functions under union and
vectors. Notions of separating conjunction as convolution arises in
all these examples in a natural way. In the disjoint union and vector
examples, the relationship between convolution, separation and
concurrency becomes most apparent.  Previously, this observation of
separating conjunction as a notion of concurrency with a strongly
spatial meaning has been one of the motivations for concurrent
separation logic~\cite{COY07} and concurrent Kleene
algebra~\cite{HMSW11}.

As a preparation we show how multisets with multiset union
and sets with disjoint union arise in the power series
setting.

\begin{example}[Multisets]\label{ex:multisets}
  Let $S$ be a set and let $f:S\to\mathbb{N}$ assign a multiplicity to
  elements of $S$. Consider the max/min-plus algebra over
  $\mathbb{N}$~\cite{GondranMinoux}, which forms a commutative
  distributive quantale. Define, rather artificially, a partial
  semigroup on $S$ by stipulating
  \begin{equation*}
    x\ast y=
    \begin{cases}
      x, &\text{ if } x=y,\\
      \bot, &\text{ otherwise}.
    \end{cases}
  \end{equation*}
  Then $\mathbb{N}^S$ is the set of multisets over the set $S$ which,
  by Theorem~\ref{thm:quantale-lifting}, forms a commutative
  distributive quantale under the operations
  \begin{gather*}
    (f\uplus g)\, x = (f\ast g)\, x= \sum_{x=x\ast x} f\, x + g\, x = f\, x + g\, x,\\
    (\sum_{i\in I}f_i)\, x = \max_{i\in I}(f_i\, x),\qquad
(\prod_{i\in I}f_i)\, x = \min_{i\in I}(f_i\, x).
  \end{gather*}
  The ``convolution'' $\uplus$ is the usual multiset addition. For
  example,
\begin{align*}
 a^2b^5c\uplus ab^3d^2 &=a^3b^8cd^2,\\
a^2b^5c+ ab^3d^2 &=a^2b^5cd^2,\\
a^2b^5c\sqcap ab^3d^2 &=ab^3. 
\end{align*}
\end{example}

\begin{example}[Powersets]\label{ex:powersets}
  Under the same conditions as in Example~\ref{ex:multisets}, suppose
  that $f:S\to\mathbb{B}$ is the characteristic function which
  determines the subsets of $S$. Then $\mathbb{B}^S\cong 2^S$ reduces
  to the complete distributive lattice of powersets of $S$; the ring
  of sets over $S$. In particular, $f\uplus g= \max(f,g)$. This
  lifting implements the powerset functor. \qed
\end{example}

Theorem~\ref{thm:quantale-lifting} shows that the function space $Q^S$
from a partial commutative semigroup $S$ into a commutative quantale
$Q$ forms a commutative quantale. In addition, we have seen in
Section~\ref{sec:booleancase}, that, in that case, $\mathbb{B}^S$ may
yield the quantale of resource predicates in which convolution is
separating conjunction. We now discuss four special cases of
separating conjunction.

\begin{example}[Separating Conjunction on Multisets]\label{ex:separating-conjunction-ms}
  The free commutative monoid $(\Al^\ast,\ast,0)$ generated by the
  alphabet $\Al$ is isomorphic to the set of all multisets over $\Al$
  with $\ast$ being multiset addition $\uplus$. By
  Theorem~\ref{thm:quantale-lifting}, $Q^{\Al^\ast}$ forms a
  commutative quantale if $Q$ does; distributivity and unitality lift
  as usual.

  Convolution $ (f\ast g)\, x =\sum_{x=y\ast z}f\, y\ast g\, z$ separates
  the multiset or resource $x$ in all possible ways and then applies
  the functions $f$ and $g$ to the result, depending on the
  interpretation of multiplication in $Q$. For $Q=\mathbb{B}$,
  $\mathbb{B}^{\Al^\ast}$ forms the resource predicate quantale over
  multisets. Convolution $f\ast g$ is separating conjunction as a
  complex product on sets
  of multisets based on multiset addition as a separator:
  \begin{equation*}
    (f\ast g)\, x =\sum_{x=y\uplus z} f\, y\sqcap g\, z.
  \end{equation*}
  \qed
\end{example}
In many contexts, multisets form a paradigmatic data type for resources.

\begin{example}[Separating Conjunction on Sets]\label{ex:separating-conjunction-sets}
  The free commutative idempotent monoid $(\Al^\ast,\ast,0)$ generated
  by the alphabet $\Al$ is isomorphic to $2^\Al$ with $\ast$ being
  union. More interesting in our context is the consideration of
  disjoint union, which is defined as
  \begin{equation*}
    x\oplus y =
    \begin{cases}
      x\cup y, & \text{ if } x\cap y=0,\\
      \bot, & \text{ otherwise}.
    \end{cases}
  \end{equation*}
  Then $(\Al^\ast,\oplus,0,\bot)$ forms a partial commutative monoid
  and, by Theorem~\ref{thm:quantale-lifting}, $Q^{\Al^\ast}$ forms a
  commutative quantale.  Convolution $(f\ast g)\, x$ now separates the
  set $x$ into disjoint subsets and then applies the functions $f$ and
  $g$ to these subsets, depending on the interpretation of $\ast$ in
  the target quantale.  For target quantale $\mathbb{B}$ we obtain the
  resource predicate quantale $\mathbb{B}^{\Al^\ast}$ on power sets
  based on disjoint union as a separator:
  \begin{equation*}
    (f\ast g)\, x =\sum_{x=y\oplus z} f\, y\sqcap g\, z.
  \end{equation*}
\qed
\end{example}
This kind of separating conjunction is particularly appropriate for
(indexed) families.

\begin{example}[Separating Conjunction on Heaplets]\label{ex:separating-conjunction-heaplets}
  Let $(S,\ast,0)$ be the partial commutative monoid of partial
  functions $\eta:A\to B$ with empty function $0:A\to B$ and
  composition defined by
    \begin{equation*}
      \eta_1\ast \eta_2 =
      \begin{cases}
        \eta_1\cup \eta_2,& \text{ if } \mathit{dom}(\eta_1)\cap\mathit{dom}(\eta_2)=\emptyset,\\
        \bot, & \text{ otherwise}.
      \end{cases}
    \end{equation*}
    The functions $\eta$ are sometimes called \emph{heaplets} and used
    to model a memory heap. As usual, by
    Theorem~\ref{thm:quantale-lifting}, $Q^S$ forms a commutative
    distributive unital quantale whenever $Q$ does. In particular,
    $\mathbb{B}^S$ forms an algebra of heap assertions with
    convolution as separating conjunction over the heap.\qed
\end{example}

\begin{example}[Separating Conjunction on Vectors]\label{ex:separating-conjunction-vectors}
  Consider a set $S$ of vectors $x$ of fixed dimension $|x|= n$. We
  turn this into a partial commutative semigroup by defining composition as
  \begin{equation*}
    (x\ast y)_i=
    \begin{cases}
      x_i,&\text{ if } y_i=0,\\
      y_i,&\text{ if } x_i=0,\\
      \bot,& \text{ otherwise}.
    \end{cases}
  \end{equation*}
  Also let $x=\bot$ if $x_i=\bot$ for some $1\le i\le n$. It is
  obvious from this definition that the zero vector $0$ is a unit with
  respect to $\ast$. For example,
  \begin{equation*}
    \begin{pmatrix}
      5 \\ 0 \\ 7
    \end{pmatrix}
\ast
\begin{pmatrix}
  0 \\ 4 \\ 0
\end{pmatrix}
=
\begin{pmatrix}
  5 \\ 4 \\ 7
\end{pmatrix}
\qquad\qquad
    \begin{pmatrix}
      5 \\ 0 \\ 7
    \end{pmatrix}
\ast
\begin{pmatrix}
  0 \\ 4 \\ 4
\end{pmatrix}
=
\bot
  \end{equation*}
  Then Theorem~\ref{thm:quantale-lifting} implies that $Q^S$ forms a
  commutative distributive unital quantale whenever $Q$ does, and
  $\mathbb{B}^S$ forms an assertion algebra with a vector-based notion
  of separating conjunction.\qed
\end{example}
The notion of separation on vectors, which splits vectors into
disjoint blocks, lends itself to transforming such vectors in parallel
fashion. This is further elaborated in Example~\ref{ex:lin-trafos}.

In separation logic, a magic wand operation is often used. It is the
upper adjoint of separating conjunction. In the quantale setting, this
adjoint exists because separating conjunction distributes over
arbitrary suprema by definition.

Additional notions of resource monoids and liftings to assertion
algebras have been studied within the Views
framework~\cite{D-YBGPY13}. Whether their generic soundness results for
Hoare logics can be reconstructed in the power series setting is left
for future work.


\section{Transformers and Bi-Quantales}\label{sec:transformers}

The powerset lifting discussed in Section~\ref{sec:fpsquantale}
suggests that state and predicate transformers could be modelled as
power series as well.  This section sketches how this can be
achieved. A detailed analysis and the consideration of particular
classes of predicate transformers is left for future work.

A \emph{state transformer} $f_R:A\to 2^B$ is often associated with a
relation $R\subseteq A\times B$ by defining
\begin{equation*}
  f_R\, a=\{b \mid (a,b)\in R\}.
\end{equation*}
State transformers are turned into \emph{predicate transformers}
$\hat{f}_R:2^B\to 2^A$ by the Kleisli lifting
\begin{equation*}
  \hat{f}_R\, Y=\{x \mid f_R\, x\subseteq Y\}. 
\end{equation*}
The following results are well known~\cite{BvW99-book}.
\begin{proposition}
The  state transformers in $(2^B)^A$ and the predicate transformers in
$(2^A)^{2^B}$ form complete distributive lattices.
\end{proposition}
\begin{proof} $2^B\cong\mathbb{B}^B$ forms a complete distributive
  lattice by Lemma~\ref{lem:lattice-lifting} because $\mathbb{B}$
  forms a complete distributive lattice. The same argument applies to
  $2^A$. It therefore follows that $(2^B)^A$ and $(2^A)^{2^B}$ are again
  complete distributive lattices by
  Lemma~\ref{lem:lattice-lifting}.
\end{proof}
Predicate transformers of type $2^A\to 2^A$ form a monoid with respect
to function composition. It is also well known that the subalgebra of
\emph{completely additive} predicate transformers, which satisfy $f\,
(\sum_{i\in I}X_i)= \sum_{i\in I} (f\, X_i)$, forms a distributive
unital quantale in which the identity function is the multiplicative
unit. However, the operation of infimum in this algebra is not the one
that is lifted pointwise; instead it is induced by the operation of
supremum~\cite{BvW99-book}. A dual result holds for \emph{completely
  multiplicative} predicate transformers, which satisfy $f\,
(\prod_{i\in I}X_i)= \prod_{i\in I} (f\, X_i)$. In this case, the
monoidal part of the quantale lifting is not obtained with the power
series lifting technique either.

The cases of resource monoids, where assertion algebras contain a
notion of separating conjunction, are more interesting.

Let $S$ be a partial monoid. A \emph{monoid transformer} is a function
of type $S\to 2^S$. A \emph{monoid predicate transformer} is a
function of type $2^S\to 2^S$. Examples are \emph{resource
  transformers} and \emph{resource predicate transformers}, in which
case $S$ is a resource monoid. Such transformers have been studied in
the context of abstract separation logic \cite{COY07}. The following
results follow immediately in our setting.
\begin{proposition}\label{prop:ptquantale}
  Let $S$ be a partial monoid. Then the monoid transformers in
  $(2^S)^S$ and the monoid predicate transformers in $(2^S)^{2^S}$
  form distributive unital quantales.  In both cases, commutativity
  lifts from $S$.
\end{proposition}
\begin{proof}
  $2^S$ forms a distributive unital quantale according to
  Corollary~\ref{cor:powerset-lifting}. It is commutative whenever $S$
  is. Hence $(2^S)^S$ forms a distributive unital quantale by
  Theorem~\ref{thm:quantale-lifting}. Commutativity lifts again from
  $S$.

  Similarly, $(2^S)^{2^S}$ is a distributive unital quantale by
  Theorem~\ref{thm:quantale-lifting} because $2^S$ is and the
  multiplicative reduct of $2^S$ is a monoid.  Commutativity lifts
  again from $S$.
\end{proof}

Proposition~\ref{prop:ptquantale} can be combined with the previous
observation about predicate transformer quantales.

\begin{theorem}\label{thm:ptbiquantale}
  Let $S$ be a partial (commutative) monoid. Then
  $((2^S)^{2^S},\subseteq,\cdot,\circ,\mathit{id},\unit)$ forms weak a
  unital bi-quantale with (commutative) convolution $\cdot$ and
  function composition $\circ$ as well as the unit function
  $\mathit{id}$ and unit power series $\unit$.
\end{theorem}
In this context, \emph{weak} means that the left distributivity law
$f\circ \sum_{i\in I} g_i = \sum_{i\in I} f\circ g_i$ need not hold in
the space of predicate transformers. It holds, however, when predicate
transformers are completely additive.


\section{Partial Power Series Quantales} 
\label{sec:partial-formal-power}

This section generalises Theorem~\ref{thm:quantale-lifting} to
situations in which the target algebras $Q$ are assumed to be partial
quantales in the sense that their semigroup retracts are partial. In
this case, partiality of composition shows up not only in the
splitting $x=y\cdot z$, but also in the product $f\, y\cdot g\, z$ in
convolutions. It turns out that the quantale structure of the target
algebra is preserved at the level of the function space, but the loss
of totality in $f\, y\cdot g\, z$ causes the function space to be
partial as well. Previous proofs must therefore be reconsidered.

As an example we consider linear transformations of vectors
implemented by matrices, in which vectors that are separated as in
Example~\ref{ex:separating-conjunction-vectors} can be transformed in
concurrent fashion by matrices which can be separated into non-zero
blocks along the diagonal.  This is a particular manifestation of the
correspondence between separation and concurrency in the context of
convolution.
\begin{proposition}\label{prop:partial-quantale-lifting}
  Let $(S,\cdot)$ be a partial semigroup. If $(Q,\le,\cdot)$ is a
  (distributive) partial quantale, then so is
  $(Q^S,\le,\cdot)$. In addition, commutativity lifts
  from $S$ and $Q$ to $Q^S$ and unitality lifts if $S$ is a partial
  monoid.
\end{proposition}
\begin{proof}
  By Theorem~\ref{thm:quantale-lifting}, the (commutative) monoidal
  and distributivity laws need to be checked.

  Suppose $(f\cdot (g \cdot h))\, x$ is defined. Then
\begin{equation*}
  (f\cdot (g \cdot h))\, x = \sum_{x=x_1\cdot (x_2\cdot x_3)} f\, x_1 \cdot (g\, x_2 \cdot h\, x_3).
\end{equation*}
Thus $x_1\cdot (x_2\cdot x_3)$ is defined and equal to $(x_1\cdot
x_2)\cdot x_3$ and $f\, x_1 \cdot (g\, x_2 \cdot h\, x_3)$ is defined and
equal to $(f\, x_1 \cdot g\, x_2) \cdot h\, x_3$. Hence
\begin{equation*}
  \sum_{x=x_1\cdot (x_2\cdot x_3)} f\, x_1 \cdot (g\, x_2 \cdot h\, x_3)=
\sum_{x=(x_1\cdot x_2)\cdot x_3} (f\, x_1 \cdot g\, x_2) \cdot h\, x_3
 = ((f\cdot g)\cdot h)\, x.
\end{equation*}
The situation where $((f\cdot g)\cdot h)\, x$ is defined is
opposition dual. Hence $Q^S$ forms a partial semigroup. 

Suppose that $ (f\cdot \sum_{i\in I} g_i)\, x$ is defined. Then
\begin{equation*}
  (f\cdot \sum_{i\in I} g_i)\, x = \sum_{x=y\cdot z} f\, y \cdot
  (\sum_{i\in I} g_i)\, z
  = \sum_{x=y\cdot z} \sum_{i\in I} (f\, y \cdot g_i\, z)
  = \sum_{i\in I} (f\cdot g_i)\, x.
\end{equation*}
The proof can be reversed if the $(f\cdot g_i)\ x$  are
defined. The proof of right distributivity is opposition dual.  This
shows that $Q^S$ forms a partial distributive quantale.

Suppose $(f \cdot g)\, x$ is defined and $S$ and $Q$ are both
commutative. Then
\begin{equation*}
  (f\cdot g)\, x = \sum_{x=y\cdot z}f\, y \cdot g\, z = \sum_{x=z\cdot y}g\, z \cdot f\, y = (g\cdot f)\, x.
\end{equation*}
This lifts commutativity.

Finally, assume that $S$ is a monoid and $Q$ is unital and define the
power series $\unit$ as usual. Suppose that $(\unit \cdot f)\, x$ is
defined. Then
\begin{equation*}
  (\unit\cdot f)\, x = \sum_{x=y\cdot z} \unit\, y \cdot f\, z = 1\cdot f\, x = f\, x.
\end{equation*}
Moreover, $f\cdot \unit =f$ follows from opposition duality. This lifts
unitality.
\end{proof}

\begin{example}[Linear Transformations of Vectors]\label{ex:lin-trafos}
  Consider again the partial semigroup $(S,\ast)$ on $n$-dimensional
  vectors from Example~\ref{ex:separating-conjunction-vectors}. It is
  easy to check that $S$ actually forms a partial commutative dioid
  with respect to $\ast$ as multiplication and standard vector
  addition. Distributivity $x\ast(y+z)=(x\ast y)+(x\ast z)$ follows
  immediately from the definition: the case of $x_i=0$ holds
  trivially, the case of $(y+z)_i=0$ requires that $y_i=z_i=0$.

  Proposition~\ref{prop:partial-quantale-lifting} then implies as a
  special case that the functions of type $S\to S$ form a commutative
  dioid; they form a trioid with the other multiplication being
  function composition. The sum in the convolution is obviously finite
  since there are only finitely many ways of splitting a vector of
  finite dimension.  In addition, the functions $f$ and $g$ in a
  convolution are not only applied to separate parts $y$ and $z$ of
  vector $x$, but they must map to separate parts $f\, y$ and $g\, z$ of
  the resulting vector as well.

  Unitality cannot be lifted as in
  Proposition~\ref{prop:partial-quantale-lifting} because the units of
  $+$ and $\ast$ coincide. It is easy to check that the unit with
  respect of $\ast$ on $S^S$ is defined as
\begin{equation*}
  e\, x = 
  \begin{cases}
    0, & \text{ if } x=0,\\
    \bot, & \text{ otherwise}.
  \end{cases}
\end{equation*}

For further illustration consider the linear transformations on
$n$-dimensional vectors given by multiplying $n$-dimensional vectors
with an $n\times n$ matrix and adding an $n$-dimensional vector.

As a simple example of a term contributing to a convolution consider
\begin{equation*}
  \begin{pmatrix}
    a_1 & b_1\\
    c_1 & d_1
  \end{pmatrix}
  \begin{pmatrix}
    x\\
    0
  \end{pmatrix}
\ast
  \begin{pmatrix}
    a_2 & b_2\\
    c_2 & d_2
  \end{pmatrix}
  \begin{pmatrix}
    0\\
    y
  \end{pmatrix}
=
\begin{pmatrix}
  a_1x\\
  c_1y
\end{pmatrix}
\ast
\begin{pmatrix}
  b_2y\\
  d_2y
\end{pmatrix}
= \bot,
\end{equation*}
whereas
\begin{equation*}
  \begin{pmatrix}
    a_1 & b_1\\
    0 & d_1
  \end{pmatrix}
  \begin{pmatrix}
    x\\
    0
  \end{pmatrix}
\ast
  \begin{pmatrix}
    a_2 & 0\\
    c_2 & d_2
  \end{pmatrix}
  \begin{pmatrix}
    0\\
    y
  \end{pmatrix}
=
\begin{pmatrix}
  a_1x\\
  0
\end{pmatrix}
\ast
\begin{pmatrix}
  0\\
  d_2y
\end{pmatrix}
= 
\begin{pmatrix}
  a_1x\\
d_2y
\end{pmatrix}.
\end{equation*}
This shows that matrices contributing to convolutions must essentially
consist of two non-trivial blocks along the diagonal modulo
(synchronised) permutations of rows and columns. That is, they are of
the form
\begin{equation*}
  \begin{pmatrix}
    M_1 & \mathbb{O}\\
\mathbb{O} & M_2,
  \end{pmatrix}
\end{equation*}
where $\mathbb{O}$ represents zero matrices of appropriate
dimension. Each pair of vectors resulting from a decomposition can be
rearranged such that the first vector consists of an upper block of
non-zero coefficients and a lower block of zeros, whereas the second
vector consists of an upper zero and a lower non-zero block, and such
that the two non-zero blocks do not overlap. One must be able to
decompose matrices and vectors of the linear transformation into the
same blocks to make convolutions non-trivial.

The transformations implemented by the above block matrix on
rearranged vectors, and more generally all linear transformations, can
clearly be executed independently or in parallel by the matrices $M_1$
and $M_2$ parts of a vector if the convolution is non-trivial.  In
this sense the convolution $\ast$ on linear transformations is a
notion of concurrent composition.\qed
\end{example}


\section{Power Series over Bi-Semigroups}
\label{sec:formal-power-series}

Our main lifting result (Theorem~\ref{thm:quantale-lifting}) shows
that the quantale structure $Q$ is preserved at the level of the
function space $Q^S$ provided that $S$ is a partial semigroup.
This 
can easily be adapted from partial semigroups $S$ to partial
$n$-semigroups and $n$-quantales with $n$ operations of composition
which may or may not be commutative. Here we restrict our attention to
bi-semigroups and bi-quantales and we discuss several examples.
\begin{proposition}\label{prop:biquantale-lifting}
  Let $(S,\circ,\bullet)$ be a partial bi-semigroup. If
  $(Q,\le,\circ,\bullet)$ is a (distributive unital) bi-quantale, then so is
  $(Q^S,\le,\circ,\bullet)$.
\end{proposition}

It is obvious that properties such as commutativity and unitality lift
as before.

 \begin{example}[Functions over Two-Dimensional Intervals]
   Closed two-dimensional intervals over a linear order can be
   defined in a straightforward way. 
   For intervals $x$ and $y$, we write $\Rectangle{x}{y}$ for the box
   consisting of points with x-coordinates in $x$ and y-coordinates in $y$.
   \begin{eqnarray*}
     \Rectangle{x}{y} & = & \{ (a,b)\ |\ a \in x \land b \in y \} \\
     \Rectangle{x}{\bot} & = & \bot \\
     \Rectangle{\bot}{y} & = & \bot
   \end{eqnarray*}

  We define the horizontal composition of two-dimensional intervals as
  \begin{equation*}
    (\Rectangle{x_1}{y_1}) \circ (\Rectangle{x_2}{y_2}) = 
      \begin{cases}
        \Rectangle{(x_1 \cdot x_2)}{y_1}, & \text{if } y_1 = y_2, \\
        \bot,                                            & \text{otherwise}.
      \end{cases}
  \end{equation*}
  and their vertical composition as
  \begin{equation*}
    (\Rectangle{x_1}{y_1}) \bullet (\Rectangle{x_2}{y_2}) = 
      \begin{cases}
        \Rectangle{x_1}{(y_1 \cdot y_2)}, & \text{if } x_1 = x_2, \\
        \bot,                                            & \text{otherwise}.
      \end{cases}
  \end{equation*}
  Whenever the target algebra forms a bi-quantale,
  Proposition~\ref{prop:biquantale-lifting} applies and the function
  space forms a bi-quantale as well. In particular, horizontal and
  vertical convolution are given by
  \begin{align*}
    (f \circ g)\, (\Rectangle{x}{y}) & = \sum_{x = x_1 \cdot x_2} f\, (\Rectangle{x_1}{y}) \circ g\, (\Rectangle{x_2}{y}), \\
    (f \bullet g)\,  (\Rectangle{x}{y}) & = \sum_{y = y_1 \cdot y_2} f\, (\Rectangle{x}{y_1}) \bullet g\, (\Rectangle{x}{y_2}).
  \end{align*}
  The situation easily generalises to n-dimensional intervals with $n$
  convolutions which may or may not be commutative.\qed
 \end{example}
 \begin{example}[Series-Parallel Pomset Languages]\label{ex:pomsets}
   Let $(S,\cdot,\ast,1)$ be a bi-monoid with non-commutative
   composition $\cdot$, commutative composition $\ast$ and shared unit
   $1$. Furthermore, let $(Q,\le,\cdot,\ast,1)$ be a bi-quantale with
   non-commutative composition $\cdot$, commutative composition $\ast$
   and shared unit $1$. Then $Q^S$ forms a bi-quantale according to
   Proposition~\ref{prop:biquantale-lifting} with a non-commutative
   convolution given by $\cdot$ and a commutative convolution given by
   $\ast$. For $\mathbb{B}^S$ and $S$ being freely generated from a
   finite alphabet $\Al$, we obtain the \emph{series-parallel pomset
     languages} or \emph{partial word languages} over $\Al$, which
   have been studied by Grabowski, Gischer and others
   \cite{Grabowski,Gischer}. They form a standard model of true
   concurrency.\qed
 \end{example}

 \begin{example}[Square Matrices with Parallel Composition]\label{ex:matrix-par}
   We define a partial commutative composition $\ast$ on square
   matrices as a generalisation of vector case, splitting matrices
   into blocks along the diagonal.
   \begin{equation*}
     (f\ast g)\, (i,j)=
     \begin{cases}
       f\ (i,j),& \text{ if } \forall k.\ g\ (i,k) = 0\wedge  g\ (k,j)= 0,\\
       g\ (i,j),& \text{ if } \forall k.\ f\ (i,k)=0\wedge g\ (k,j)=0,\\
       \bot, & \text{ otherwise}.
     \end{cases}
   \end{equation*}
   Associativity and commutativity of this operation is easy to check;
   (infinite) distributivity holds as well. It follows that square
   matrices into suitable coefficient algebras form partial
   bi-quantales.\qed
 \end{example}
 Examples~\ref{ex:pomsets} and~\ref{ex:matrix-par} thus show other
 situations where a commutative convolution gives rise to a notion of
 parallel or concurrent composition.


\section{Two-Dimensional Power Series
  Bi-Quantales}\label{sec:fpsbiquantale}

We now extend the power series approach to two dimensions; an
extension to $n$ dimensions can be obtained along the same lines.  We
consider two separate partial semigroups or monoids $(S_1,\circ)$ and
$(S_2,\bullet)$. In many cases, $S_2$ is assumed to be
commutative. This differs from \refsec{sec:formal-power-series} in
that two different semigroups algebras are lifted to a bi-quantale,
whereas in \refsec{sec:formal-power-series} a bi-semigroup is lifted
to a bi-quantale.

We consider functions $F:S_1\to S_2\to Q$ from the partial semigroups
$S_1$ and $S_2$ into an algebra $Q$, usually a bi-quantale. Note that
$A\to B\to C$ stands for $A\to (B\to C)$, and we write $(C^B)^A$ for
the class of functions of that type.

The main construction is as
follows. Theorem~\ref{thm:quantale-lifting} can be applied to
semigroup $S_1$ and target algebra $Q^{S_2}$ to lift to
$(Q^{S_2})^{S_1}$. Alternatively, $S_2$ and $Q^{S_1}$ can be lifted to
$(Q^{S_1})^{S_2}$. The algebras $Q^{S_1}$ and $Q^{S_2}$ can be
obtained by lifting as well; they can be considered as partial
evaluations of a power series $F:S_1\to S_2\to Q$ to power series
$F^y:S_1\to Q$ and $F^x:S_2\to Q$ where 
\begin{align*}
  F^y  = \lambda x.\ F\, x\, y, \qquad\qquad F^x = \lambda y.\ F\, x\, y
\end{align*}
This construction can be iterated $n$ times for power series
$F:S_1\to\dots\to S_n\to Q$.

It is well known that the function spaces obtained are isomorphic: in
general $(C^A)^B\cong (C^B)^A\cong C^{A\times B}\cong C^{B\times A}$
under the Curry-Howard isomorphism.  A categorical framework is
provided by the setting of symmetric monoidal closed categories
\cite{Kelly}, which we do not explore further in this
article. Instead we move freely between isomorphic function spaces.

By analogy to the one-dimensional case of  power series we
define operations on the function space $Q^{S_1\times S_2}$ which lift
the corresponding operations on $Q$. Ultimately our aim is to show
that bi-quantale axioms lift from $Q$ to $Q^{S_1\times S_2}$. We
define
\begin{align*}
  (\sum_{i\in I} F_i)\, x\, y &= \sum_{i\in I} (F_i\, x\, y),\\
  (\prod_{i\in I} F_i)\, x\, y &= \prod_{i\in I} (F_i\, x\, y),\\
  (F\circ G)\, x\, y &= \sum_{x=x_1\circ x_2} F\, x_1\, y\circ G\, x_2\, y,\\
  (F \bullet G)\, x\, y &= \sum_{y=y_1\bullet y_2} F\, x\, y_1\bullet G\, x\, y_2.
\end{align*}
As in the one dimensional case, $\zero = \sum_{i\in\emptyset} F_i$.  The
convolution $F\circ G$ acts on the first parameter whereas $F\bullet
G$ acts on the second one; $\sum_{i\in I}F_i$ and $\prod_{i\in I}F_i$ are defined by
pointwise lifting on both arguments. 

We now show how two-dimensional lifting results can be obtained in a
modular fashion from one-dimensional ones with
Theorem~\ref{thm:quantale-lifting}.  By currying consider the
functions $F^y:S_1\to Q$ and $F^x:S_2\to Q$. 
For these we can reuse the definitions of suprema, infima and
convolution from the one dimensional case in
Section~\ref{sec:fpsquantale}. Suprema, for instance, are given by
\begin{equation*}
  (\sum_{i\in I} F_{i}^y)\, x  =\sum_{i\in I}(F_{i}^y\, x),\qquad
  (\sum_{i\in I} F_{i}^x)\, y  =\sum_{i\in I}(F_{i}^x\, y).
\end{equation*}
The equations for infima are lattice dual.  Convolutions are given by
\begin{align*}
  (F^y\circ G^y) \, x  =\sum_{x=x_1\circ x_2}F^Y\, x_1\circ G^y\, x_2,
  \qquad
  (F^x\bullet G^x)\, y  =\sum_{y=y_1\bullet y_2}F^x\, y_1\bullet G^x\, y_2.
\end{align*}

The relationship between operations of different dimensions is
captured by the following lemma.
\begin{lemma}\label{P:spacered1}~
  The maps $\varphi_1: Q^{S_1\times S_2}\to Q^{S_2}$ and $\varphi_2:Q^{S_1\times S_2}\to
  Q^{S_1}$  defined by
  \begin{equation*}
    \varphi_1=\lambda X.\ (X)^y,\qquad \varphi_2=\lambda X.\ (X)^x
  \end{equation*}
  are homomorphisms.
\begin{enumerate}
\item $(\sum_{i\in I} F_i)^y= (\sum_{i\in I} F_{i}^y)$ and $(\sum_{i\in
    I} F_i)^x= (\sum_{i\in I} F_{i}^x)$,
\item $(\prod_{i\in I} F_i)^y= (\prod_{i\in I} F_{i}^y)$ and $(\prod_{i\in
    I} F_i)^x= (\prod_{i\in I} F_{i}^x)$,
\item $(F\circ G)^y= (F^y\circ G^y)$ and  $(F\bullet G)^x= (F^x\bullet G^x)$.
\end{enumerate}
\end{lemma}
\begin{proof}
  We only provide proofs for the first conjunct of $(a)$ and for
  $(c)$. The remaining proofs are similar.  For addition we calculate
  \begin{equation*}
    (\sum_{i\in I} F_i)^y\, x = (\sum_{i\in I} F_i)\, x\, y
    = \sum_{i\in I}(F_i\, x \, y)
    = \sum_{i\in I} (F_{i}^y\, x)
    = (\sum_{i\in I}F_{i}^y)\, x.
  \end{equation*}
  For composition $\circ$,
  \begin{align*}
   (F\circ G)^y\, x  &=(F\circ G)\, x\, y\\
  &=\sum_{x=x_1\circ x_2} (F\, x_1\, y)\circ (G\, x_2\, y)\\
&=\sum_{x=x_1\circ x_2} (F^y\, x_1)\circ (G^y\, x_2)\\
&= (F^y\circ G^y)\, x.
  \end{align*}
\end{proof}

If $(S_1,\circ,1_\circ)$ and $(S_2,\bullet,1_\bullet)$ are partial
monoids and the bi-quantale $Q$ has units $1^y$ and $1^x$ with respect
to $\circ$ and $\bullet$ (overloading notation), we define units on
$Q^{S_1\times S_2}$ as
\begin{equation*}
  \unit_\circ = \lambda x,y.
  \begin{cases}
    1_\circ, & \text{if } x = 1_\circ,\\
    0, & \text{otherwise},
  \end{cases}
\qquad
 \unit_\bullet = \lambda x,y.
  \begin{cases}
    1_\bullet, & \text{if } y = 1_\bullet,\\
    0, & \text{otherwise}.
  \end{cases}
\end{equation*}
The following result links these binary units with the unary units
$(\unit^y)_\circ:S_1\to Q$ and $(\unit^x)_\bullet:S_2\to Q$, as
defined in Section~\ref{sec:fpsquantale}.
\begin{lemma}\label{P:spacered2}~ 
  \begin{enumerate}
  \item $(\unit_\circ)^y = (\unit^y)_\circ$,
  \item  $(\unit_\bullet)^x= (\unit^x)_\bullet$.
  \end{enumerate}
\end{lemma}
\begin{proof}
  For (a),
  \begin{equation*}
    (\unit_\circ)^y\, x = \unit_\circ\, x\, y =  
    \begin{cases}
      1_\circ, & \text{ if } x=1_\circ,\\
      0, & \text{otherwise}.
    \end{cases}
    =(\unit^y)_\circ\, x.
  \end{equation*}
 The proof of (b) is similar.
\end{proof}

By Lemmas~\ref{P:spacered1} and~\ref{P:spacered2}, a lifting from $Q$
can be decomposed into a lifting to $Q^{S_2}$ and, if the lifted
property is preserved, a function application in
$(Q^{S_2})^{S_1}$. Alternatively one can lift to $Q^{S_1}$ and then
use function application in $(Q^{S_1})^{S_2}$. In the above
constructions, there are two kinds of liftings: pointwise liftings
from $Q$ to $Q^{S_1}$ or $Q^{S_2}$ and lifting by convolution for
$Q^{S_1}$ and $Q^{S_2}$.

\begin{proposition}\label{prop:seq-lifting}
  Let $(S_1,\circ)$ be a partial semigroup and $S_2$ a set.  If
  $(Q,\le,\circ)$ is a (distributive) quantale, then so is
  $(Q^{S_1\times S_2},\le,\circ)$. Unitality and commutativity lift
  from $S_1$ and $Q$ to $Q^{S_1\times S_2}$.
\end{proposition}
\begin{proof}~ If $S_1$ is a partial semigroup and $Q$ a
  (distributive) quantale, then $Q^{S_1}$ is a (distributive) quantale
  by \refthm{thm:quantale-lifting}, and by $\lambda$-abstraction for
  $F=\lambda y.\ F^y\ x$ and the homomorphic properties of $(.)^y$ in
  Lemma~\ref{P:spacered1}. For example,
  \begin{align*}
    ((F\circ G)\circ H)\, x\, y &= ((F\circ G)\circ H)^y\, x\\
& =
    ((F^y\circ G^y)\circ H^y)\, x\\
& = (F^y \circ (G^y\circ H^y))\, x\\
& = (F\circ (G\circ H))^y\, x\\
 &= (F\circ (G\circ H))\ x\ y.
   \end{align*}
  If the quantale $Q$ is unital, then so is $Q^{S_1}$, again by
  \refthm{thm:quantale-lifting}. As previously, this follows by
  $\lambda$-abstraction and the homomorphic properties of $(.)^y$ by
  Lemmas~\ref{P:spacered1} and \ref{P:spacered2}. For instance,
  \begin{equation*}
    (\unit_\circ\circ F)\, x\, y = (1_\circ\circ
    F)^y\ x = (1_\circ^y \circ F^y)\, x =  F^y\, x =F\, x\, y.
  \end{equation*}
  If $S_1$ and $Q$ are both commutative, then 
  \begin{equation*}
    (F\circ G)\, x\, y=(F\circ G)^y\, x= (F^y \circ G^y)\, x =
    (G^y\circ F^y)\, x= (G\circ F)^y\ x = (G\circ
    F)\, x\, y
  \end{equation*}
  with the homomorphism properties of $(.)^y$ and commutativity on
  $Q^{S_1}$ due to \refthm{thm:quantale-lifting}.
\end{proof}
The next statement is immediate since
$Q^{S_1\times S_2}$ and $Q^{S_2\times S_1}$ are isomorphic.
\begin{corollary}\label{cor:conc-lifting} 
  Let $S_1$ be a set and $(S_2,\bullet)$ a partial semigroup. If
  $(Q,\le,\bullet)$ is a (distributive) quantale, then so is
  $(Q^{S_1\times S_2},\le,\bullet)$. Unitality and commutativity lift
  from $S_2$  and $Q$ to $Q^{S_1\times S_2}$.
\end{corollary}
Proposition~\ref{prop:seq-lifting} and
Corollary~\ref{cor:conc-lifting} can therefore be combined into the
following lifting theorem for two-dimensional  power series.
\begin{theorem}\label{thm:biquantale}
  Let $(S_1,\circ)$ and $(S_2,\bullet)$ be partial semigroups. If
  $(Q,\le,\circ,\bullet)$ is a (distributive) bi-quantale, then so is
  $(Q^{S_1\times S_2},\le,\circ,\bullet)$. It is unital whenever $Q$
  is unital and $S_1$ and $S_2$ are partial monoids. A convolution on
  $Q^{S_1\times S_2}$ is commutative if the corresponding composition
  on $S_i$ and $Q$ are commutative.
\end{theorem}
Remember that a unital bi-quantale may have different units for its
two compositions.

As already mentioned, the construction of the bi-quantale of
two-dimensional power series generalises immediately to $n$ underlying
partial semigroups $(S_i,\circ_i)$, $n$-dimensional power series
$F:S_1\to \dots \to S_n\to Q$ and convolutions
\begin{equation*}
  (F\circ_i G)\ \dots \ x_i \ \ldots =\sum_{x_i=y\circ_i z} (F\ \dots \ y \ \dots)\circ_i (G\ \dots\ z\ \dots).
\end{equation*}
We do not pursue this generalisation in this article; the lifting
arguments apply without modification.


\section{Examples}\label{sec:biquantale-examples}

As examples of two-dimensional bi-quantales we present two interval
based models that distinguish between time and space dimensions. The
monoidal operators may be used to separate these two dimensions
independently; time is separated using chop, space using separating
conjunction as a notion of concurrent composition. The consideration
of such algebras with both kinds of separation was the starting point
of this article. In the second example of vector stream interval
functions, spatial or concurrent splitting is of course commutative,
whereas temporal splitting is not.

\begin{example}[Stream Interval
  Functions]\label{ex:stream-interval-functions}
  Let $(S_1,\cdot)$ be the partial semigroup $(I_P,\cdot)$ of closed
  intervals $I_P$ under interval function as in
  Example~\ref{ex:interval-functions} and let $S_2$ be the set of all
  functions of type $P\to A$ for an arbitrary set $A$. It follows from
  Proposition~\ref{prop:seq-lifting} that $Q^{I_P\times A^P}$ forms a
  distributive quantale, whenever $Q$ is a distributive quantale. A
  unit can be adjoined to $Q^{I_P\times A^P}$ along the lines of
  Example~\ref{ex:interval-functions}, but with a second parameter.

  As a typical interpretation, consider $P=\mathbb{R}$ with the
  standard order on reals as a model of time and let functions
  $f:\mathbb{R}\to A$ model the temporal behaviour or trajectories of
  some system. For instance, $f$ could be the solution of a
  differential equation. In that case, $F\ x\ f$ evaluates the
  behaviour of system $f$ in the interval $x$. Such kinds of functions
  have been called \emph{stream interval functions}~\cite{DHD14}. The
  convolution
  \begin{equation*}
    (F\cdot G)\, x\, f = \sum_{x=y\cdot z} (F\, y\, f) \cdot (G\, z\, f)
  \end{equation*}
  splits the interval $x$ into all possible prefix/suffix pairs $y$
  and $z$, applies $F$ to the behaviour of $f$ on interval $y$ and $G$
  to the behaviour of $f$ on interval $z$ and then combines these
  results. There are different ways in which the application of stream
  interval functions can be realised.  Moreover, the situation
  generalised to arbitrary finitely bounded intervals without fusion.

  As in the case of interval functions, our prime example of stream
  interval functions are \emph{stream interval predicates}, where
  $Q=\mathbb{B}$. Then convolution becomes a generalised version of
  chop or non-commutative separating conjunction:
 \begin{equation*}
    (F\cdot G)\, x\, f = \sum_{x=y\cdot z} (F\, y\, f) \sqcap (G\, z\, f).
  \end{equation*}
  A predicate $F$ could, for instance, test the values of a function
  $f$ over an interval $x$---at all points of $x$, at some points of
  $x$, at almost all points of $x$, at no points of $x$ and so on. It
  could, for instance, test, whether the trajectory of system $f$
  evolves within given boundaries, that is a flight path is within a
  given corridor or that a train moves according to a given time
  schedule.

  More concretely, let $P=A=\mathbb{R}$ and that $f\ t=t^3$ as shown
  below. Note that the diagram is not drawn to scale. 
  \begin{center}
    \scalebox{1.25}{\input{tcube.pspdftex}}
  \end{center}
  Let
  \begin{equation*}
    F\ x\
    f= \forall t\in x.\ f\ t\ge 0,\qquad\qquad G\ x\ f =\forall t\in x.\ f\
    t<0.
  \end{equation*}
  Then $F\ [0,10]\ f=1$ and $G\ [-7,-1]\ f=1$, but $F\ [-2,-1]\ f=0$
  and $G\ [-7,0]\ f=0$.  \qed
\end{example}

Stream interval predicates have been used to reason about real-time
systems \cite{DHD14}, but their interpretation in terms of power
series is new.  It is worth noting that $P$ may be instantiated to
other partial orders (e.g., $\mathbb{Z}$), allowing one to model both
discrete and continuous systems. 

Using Theorem~\ref{thm:biquantale}, one may further develop this
approach with rules for system-level reasoning by decomposing systems
along a time and space dimension. To the best of our knowledge, our
treatment is the first to offer both decompositions and to add a
natural notion of concurrency to interval logics. Exploration of these
rules in concrete models as well as their application towards
verification of example systems is left as future work. Here we
present one single example which is based on vectors of functions.

\begin{example}[Vector Stream Interval Functions]\label{ex:vector-stream-interval-functions}
  Let $f$ from the previous example now be a vector or product of
  functions $f_i$ such that $f:P\to A^n$, or more concretely
  $f:\mathbb{R}\to A^n$. One can then split $f(t)$ as in
  Example~\ref{ex:separating-conjunction-vectors} with respect to the
  commutative operation $\ast$ on $A^n$. For functions $f,g:P\to A^n$
  we define
  \begin{equation*}
    (f\ast g)\, p = f\, p \ast g\, p
  \end{equation*}
  by pointwise lifting. This turns $(S_2,\ast)=((A^n)^P,\ast)$ into a
  partial commutative semigroup, whereas $(S_1,\cdot)$ is again the
  partial semigroup $(I_P,\cdot)$. According to
  Theorem~\ref{thm:biquantale}, $Q^{S_1\times S_2}$ forms a
  distributive bi-quantale with commutative convolution $\ast$
  whenever $Q$ does.

  The stream interval predicates in the case of $Q=\mathbb{B}$ yield
  once more an interesting special case.  Now a vector of functions,
  for instance the solution to a system of differential equations, is
  applied to arguments ranging over an interval and the stream
  interval predicates evaluate the behaviour modelled by this vector
  of functions on the interval. 

  The convolution
  \begin{equation*}
    (F\cdot
    G)\, x\, f = \sum_{x=y\cdot z}(F\, y\, f)\sqcap (G\, z\, f)
  \end{equation*}
  can be seen as a horizontal composition. It evaluates the full vector
  of functions to splittings of the interval $x$, using $F$ for the
  prefix part of the splitting and $G$ for its suffix part.  In the
  context of interval logics this corresponds to a chop operation,
  which has a temporal flavour.
  
  The convolution or separating conjunction
  \begin{equation*}
    (F\ast G)\, x\, f= \sum_{f=g\ast h} (F\, x\, g)\sqcap (G\, x\, h)
  \end{equation*}
  can be seen as a vertical composition. It evaluates the conjunction
  of $F$ and $G$, which is obtained by separating the vector $f$ into
  all possible parts $g$ and $h$, over the full interval $x$.  Applied
  to vectors this adds an algebraic notion of concurrent composition
  to interval calculi; it clearly has a spatial flavour. 

  The two types of convolution may be distinguished using diagrams
  such as the ones below, where time occupies the $x$-axis and space
  the $y$-axis. 
  \begin{center}
    \scalebox{0.75}{\input{mult-2dim.pspdftex}}
  \end{center}
  The left diagram depicts $(F \ast G) \cdot (H \ast K)$, where the
  convolution first splits the stream interval function along the
  $x$-axis (time dimension) to give us formulae $F \ast G$ and $H \ast
  K$. Each of these is then split along the $y$-axis (space
  dimension). On the other hand the right diagram depicts $(F \cdot H)
  \ast (G \cdot K)$, where the space dimension is split first to give
  $F \cdot H$ and $G \cdot K$, followed by a split along the time
  dimension.  \qed
\end{example}
The examples in Section~\ref{sec:fpscomquantaleexamples} suggest that
other notions of spatial separation, for instance those based on
disjoint unions for families of functions, or more specific notions
such as separating conjunction on heaps, can be used instead of vector
separation.  Theorem~\ref{thm:biquantale} is modular in this
regard. We therefore do not present these examples in detail.


\section{Power Series over Futuristic Monoids}\label{sec:futuristic}

This section adapts the power series approach to a case which is
appropriate, for instance, for languages with finite and infinite
words and for intervals which may be semi-infinite in the sense that
they have no upper bounds.  Such approaches are, for instance,
appropriate for total correctness reasoning, where termination cannot
be assumed or for reactive (concurrent) systems. We model these cases
abstractly with monoids which, due to lack of better nomenclature, we
call futuristic.

Formally, a partial semigroup $(S,\cdot)$ is \emph{futuristic} if
$S=S^u\cup S^b$, $S^u\cap S^b=\emptyset$ and $x\cdot y$ is undefined
whenever $x\in S^u$. Thus, $S^u$ and $S^b$ correspond to the unbounded
and bounded elements of $S$, respectively. For $S^b$, we require that if
$x \cdot y \in S^b$, then $x \in S^b$.

In that case, for $f,g:X\to Y$, we define
\begin{equation*}
  (f\cdot g)\, x =  \sum_{x=y\cdot z} (f\, y)\cdot (g\, z) +
  \begin{cases}
    f\, x, &\text{ if } x\in S^u,\\
    0, & \text{ if } x\in S^b.
  \end{cases}
\end{equation*}

\begin{lemma}
  \label{lem:futuristic-quantale}
  Let $(S,\cdot)$ be a futuristic partial semigroup. If $Q$ is a
  (distributive) quantale, then $Q^S$ is a (distributive) quantale
  with $\zero:S\to Q$ not necessarily a right annihilator and left
  distributivity holding only for non-empty suprema.
\end{lemma}
\begin{proof}
  We need to verify the laws involving `$\cdot$' with our new
  multiplication. It suffices to consider the cases where $x\in S^u$;
  the others are covered by Theorem~\ref{thm:quantale-lifting}. For
  left distributivity we calculate, for $I\neq\emptyset$, 
  \begin{align*}
    (f\cdot \sum_{i\in I}g_i)\, x & = f\, x + \sum_{x=y\cdot z} f\, y\cdot \sum_{i\in I}(g_i\, z)\\ 
&= (\sum_{i\in I} f\, x) + \sum_{i\in I}\sum_{x=y\cdot z} (f\, y\cdot g_i\, z)\\
&= \sum_{i\in I} (f\, x + \sum_{x=y\cdot z} (f\, y\cdot g_i\, z))\\
&= (\sum_{i\in I} (f\cdot g_i))\, x.
  \end{align*}
  For $I=\emptyset$, however $(f\cdot \zero)\, x = f\, x$ if $x\in
  S^u$, hence in this case left distributivity fails.

  For right distributivity, which is no longer opposition dual, we
  calculate
  \begin{align*}
    ((\sum_{i\in I}f_i)\cdot g)\, x &= (\sum_{i\in I} f_i\, x) + \sum_{x=y\cdot z} (\sum_{i\in I} f_i\, y)\cdot g\, z\\
&= (\sum_{i\in I} f_i\, x) + \sum_{i\in I}\sum_{x=y\cdot z} (f_i\, y\cdot g\, z)\\
&= \sum_{i\in I} (f_i\, x + \sum_{x=y\cdot z} (f_i\, y\cdot g\, z))\\
&=(\sum_{i\in I} (f_i\cdot g))\, x.
  \end{align*}
  Left annihilation is as usual a special case of right
  distributivity. We calculate explicitly
\begin{equation*}
  (\zero\cdot f)\, x = \zero\, x + \sum_{x=y\cdot z} \zero\,
  y\cdot f\, z= 0+0=0 
\end{equation*}
Finally, for associativity, we calculate
    \begin{align*}
      (f\cdot (g\cdot h))\,x 
&= f\,x + \sum_{x=y\cdot z} f\,y\cdot (g\,z + \sum_{z=u\cdot v} g\,u\cdot h\,v)\\
&= f\,x + (\sum_{x=y\cdot z} f\,y\cdot g\,z) + \sum_{x = y \cdot z} f\, y \cdot (\sum_{z=u \cdot v} g\,u\cdot h\,v)\\
&= (f\cdot g)\,x + \sum_{x=y \cdot u \cdot v} f\,y\cdot g\,u \cdot h\,v\\
&= (f\cdot g)\,x + \sum_{x=w \cdot v} (\sum_{w = y \cdot u} f\,y \cdot g\,u) \cdot h\,v\\
&= (f\cdot g)\,x + \sum_{x=w\cdot v}(f\cdot g)\,w \cdot h\,v\\
&= ((f\cdot g)\cdot h)\,x.
    \end{align*}
The last but first step uses the fact that $w\in S^b$
\end{proof}
\begin{proposition}\label{prop:futuristic-biquantale}
  Let $(S_1,\circ)$ be a futuristic partial semigroup and $S_2$ a set.
  If $(Q,\le, \circ)$ is a (distributive) quantale, then
  $(Q^{S_1\times S_2},\le,\circ)$ is a (distributive) quantale with
  $\zero$ not necessarily a right annihilator and left distributivity
  holding only for non-empty suprema. Unitality lifts from $Q$ to
  $Q^{S_1\times S_2}$ with unit $\unit_\circ$ if $S_1$ is a partial
  monoid.
\end{proposition}
The proof adapts that of Proposition~\ref{prop:seq-lifting} to
Lemma~\ref{lem:futuristic-quantale}. A treatment of historistic
intervals is dual, that is, left annihilation
fails. Proposition~\ref{prop:futuristic-biquantale} can be extended
further into an analogue of Theorem~\ref{thm:biquantale}. We do not
explicitly display this statement.

\begin{example}[Formal Languages with Infinite Words]
  Let $\Al$ be a finite alphabet. Let $\Al^\ast$, as previously,
  denote the set of finite words over $\Al$ and $\Al^\omega$ the
  set of all infinite words, which are sequences of type
  $\mathbb{N}\to \Al$. Let
  $\Al^\infty=\Al^\ast\cup\Al^\omega$. Then $\Al^\ast\cap
  \Al^\omega=\emptyset$ by definition. Every language
  $L\subseteq\Al^\infty$ may contain finite as well as infinite
  words and we write $\mathsf{fin}(L)$ and $\mathsf{inf}(L)$ for
  the sets of all finite and infinite words in $L$.

  In this context it is natural to disallow the concatenation of an
  infinite word with another word, hence $\Al^\infty$ is endowed
  with a futuristic partial monoid structure. In addition, the product
  of $L_1,L_2\subseteq\Al^\infty$ is commonly defined as
  \begin{equation*}
    L_1\cdot L_2=\mathsf{inf}(L_1)\cup \{vw \mid v\in\mathsf{fin}(L_2)\wedge w\in L_2\}.
  \end{equation*}
  This is captured by the futuristic product with $Y=\mathbb{B}$. It
  then follows from Lemma~\ref{lem:futuristic-quantale} that
  $\Al^\infty$ forms a distributive quantale in which
  $L\cdot\emptyset=\emptyset$ need not hold and left distributivity
  holds only for non-empty suprema. In fact, the absence of right
  annihilation can be verified with the singleton stream
  $L=\{aaa\dots\}$.\qed
\end{example}
Models with finite/infinite paths and traces can be built in a similar
fashion.

\begin{example}[Functions and Predicates over Futuristic Intervals]
  Let $(P,\le)$ be a linear order without right endpoint. Let $I_P^f$
  stand for the set of all non-empty closed intervals over $X$ and let
  $I_X^i$ denote the set of all \emph{futuristic intervals}
  $[a,\infty]= \{ b\ |\ b\ge a\}$. This does not mean that we add an
  explicit element $\infty$ to $X$; $\infty$ is merely part of our
  naming conventions. Then $I_X=I_X^f\cup I_X^i$ and $I_X^f\cap
  I_X^i=\emptyset$. The fusion product of intervals can now be
  redefined as
\begin{equation*}
  x\cdot y =
  \begin{cases}
    x, &\text{ if } x\in I_X^i,\\
    [x_{\min},y_{\max}], &\text{if } x\in I_X^f \text{ and } x_{\max}=y_{\min},\\
\bot, &\text{otherwise},
  \end{cases}
\end{equation*}
where $y_{\max}=\infty$ is included as an option.  It then follows
from Lemma~\ref{lem:futuristic-quantale} that $Q^{I_P}$ forms a
distributive quantale in which $\zero$ is not necessarily a right
annihilator. In fact, $f\circ \zero = \zero$ can be falsified with any
interval $x=[a,\infty]$ and interval predicate $f=\lambda x.\ a\in
x$.\qed
\end{example}
An example of closed and open intervals without fusion can be obtained
along the same lines. Examples of bi-quantales based on stream
functions over futuristic intervals with a notion of separating
conjunction can be obtained in a straightforward way.


\section{Interchange Laws}\label{sec:interchange}

Algebras in which a spatial or concurrent separation operation
interact with a temporal or sequential one have already been studied,
for instance, in the context of concurrent Kleene
algebra~\cite{HMSW11}. In addition to the trioid or bi-quantale laws,
these algebras provide interesting interaction laws between the two
compositions, which in this context are interpreted as concurrent and
sequential composition.  Such laws are, obviously, of
general interest.

More concretely, the following \emph{interchange laws} hold in
concurrent Kleene algebras:
\begin{align*}
  (x\ast y)\cdot z & \le x\ast (y\cdot z),\\
  x\cdot (y\ast z) & \le (x\cdot y)\ast z ,\\
  (w\cdot x)\ast (y\cdot z) & \le (w\cdot y)\ast (x\cdot z).
\end{align*}
We call the first two laws \emph{small interchange laws} and the last
one \emph{weak interchange law}.  These laws hold in models of
concurrency including shuffle languages and certain classes of
partially ordered multisets~\cite{Gischer}. It has been shown that one
of the small interchange laws is equivalent to a separation logic style
frame rule in a certain encoding of Hoare logics~\cite{Locality}. The
weak interchange law, in turn, is equivalent to one of the standard
concurrency rules for Hoare logic, which is similar to those
considered in Owicki and Gries' logic~\cite{OG76} or in concurrent
separation logic~\cite{COY07}. This relationship is considered further
in Section~\ref{sec:hoare}.

The close relationship between power series and separation logic and
the similarity between two-dimensional power series and concurrent
Kleene algebras make it worth considering the interchange laws in this
setting.  However we obtain mainly negative results.

To start with a positive result, we establish interchange laws between
other kinds of operations.



\begin{lemma}\label{P:quantaleprops}
  In every quantale, the following interchange laws hold:
\begin{equation*}
  (w\sqcap x)\cdot (y\sqcap z) \le (w\cdot
  y)\sqcap (x\cdot z),\qquad (w\sqcap x)\ast (y\sqcap z) = (w\ast
  y)\sqcap (x\ast z).
  \end{equation*}
\end{lemma}

It turns out, however, that the small and weak interchange laws
between sequential and concurrent composition do not hold in
general. This is established by the counterexamples which support the
following lemma.
\begin{proposition}
  There are $F,G,H,K:S_1\to S_2\to \mathbb{B}$ such that the following
  holds.
  \begin{enumerate}\label{prop:interchangeref}
  \item $F\cdot G\not\le F\ast G$,
  \item $(F\ast G)\cdot H\not\le F\ast (G\cdot H)$,
  \item $F\cdot (G\ast H)\not\le (F\cdot G)\ast H$,
  \item $(F\ast G)\cdot (H\ast K)\not\le (F\cdot H)\ast (G\cdot K)$.
  \end{enumerate}
\end{proposition}
\begin{proof}
  First, note that $\le$ can be interpreted as $\Rightarrow$ for
  stream interval predicates, and recall that parallel composition of
  predicates is separating conjunction when $f$ is a vector of
  functions. 
\begin{enumerate}
\item To refute $F\cdot G\le F\ast G$, let $x=[-10,10]$, $f=(f_1,f_2)$
  with
  \begin{equation*}
    f_1\, t =
    \begin{cases}
      1, &t\le 0,\\
      0, &t> 0,
    \end{cases}
    \qquad\qquad
    f_2\, t = 
    \begin{cases}
      0, &t\ge 0,\\
      1, &t<0,
          \end{cases}
  \end{equation*}
  and
  \begin{equation*}
    F\, x \, f = \forall t \in x.\ f_1\, t = 1,\qquad G\, x \, f = \forall
  t\in x.\ f_2\, t = 1.
\end{equation*} 
Then $(F\cdot G)\, x\, f =1$, splitting interval $x$ at $t=0$, whereas
$(F \ast G)\, x\, f =0$ since neither $F$ nor $G$ holds on the entire
interval $x$. This may be visualised using the diagrams below, where
dashed lines represent that the corresponding function has value $0$,
and solid lines represent a value $1$. For the right diagram, there is
not possible way for the vectors $f_1$ and $f_2$ to go through $F$ and
$G$. 
\begin{center}
  \scalebox{0.75}{\input{no-weak-inter-3.pspdftex}}
\end{center}
\item To refute $(F\ast G)\cdot H\le F\ast (G\cdot H)$, let $x=[-10,10]$,
  $f_1$ as in (a) and $f_2=\lambda t.\ 0$, where
  \begin{align*}
    F\, f\, x& =  \forall t\in x.\ f_1\, t = 1, \\
    G\, f\, x &=   \forall t\in x.\ f_2\, t=0, \\
    H\, f\, x &=    \forall t\in x.\ f_1\, t=0\vee f_2\, t =0.
  \end{align*}
  This makes the left hand side  $1$ and the right hand side $0$. This
  is visualised by the diagram below---neither $f_1$ nor $f_2$ may
  go through $F$. 
  
\begin{center}
  \scalebox{0.75}{\input{no-weak-inter-4.pspdftex}}
\end{center}

\item $H\cdot (G\ast F)\le (H\cdot G)\ast F$ can be
refuted by function
\begin{equation*}
  f_1'\, t =
  \begin{cases}
    0,& t\le 0,\\
    1,& t >0,
  \end{cases}
\end{equation*}
and $f_2$ as in (b), exploiting opposition duality between the two
interchange laws and realising that $f_1'$ is the ``time reverse'' of
$f_1$.

\item To refute $(F\ast G)\cdot (H\ast K)\le (F\cdot H)\ast (G\cdot K)$,
consider $f = (f_1, f_2, f_3)$ where
\begin{equation*}
  f_1\, t  =  0,\qquad
  f_2\, t  = 
  \begin{cases}
    0,& t\le 0,\\
    1,& t >0,
  \end{cases}
  \qquad
  f_3\, t  =  1 
\end{equation*}
and
\begin{align*}
  F\, f\, x  &=  \forall t \in x.\ f_1\, t = 0, \\
  G\, f\, x  &= \forall t \in x.\ f_2\, t < f_3\, t, \\
  H\, f\, x &=  \forall t \in x.\ f_1\, t < f_2\, t, \\
  K\, f\, x  &=  \forall t \in x.\ f_3\, t = 1. 
\end{align*}
For $x = [-10, 10]$, the diagram on the left below shows that the left
hand side $(F\ast G)\cdot (H\ast K)$ holds. However, in the diagram on
the right, which represents $(F\cdot H) \ast (G\cdot K)$, there is no
possible combination of horizontal and vertical splits that satisfy
$f$. In particular, $f_1$ must go through $F$, and similarly $f_3$
must go through $K$. We have a choice of placing $f_2$ above the
horizontal line (through $F$ and $H$), or below (through $G$ and $K$),
however, neither choice is appropriate.
\begin{center}
   \scalebox{0.75}{\input{no-weak-inter-2.pspdftex}}
\end{center}
\end{enumerate}
\end{proof}

Imposing addition algebraic restrictions, which would allow the
derivation of interchange laws, is left for future work. A promising
candidate is the consideration of locality assumptions, as in
separation logic~\cite{COY07}, which are briefly explained in the
following section, or the inclusion of dependency
relations~\cite{HMSW11} in the definition of the semigroup operations.


\section{Hoare Logics from Power Series Quantales}\label{sec:hoare}

One benefit of algebras is that they support the development of
verification systems. It is well known, for instance, that quantales
can be endowed with Hoare logics~\cite{HMSW11}, more precisely
\emph{propositional} Hoare logics, in which data flow rules such as
assignment rules are missing. This section illustrates how this leads to
propositional Hoare logics over power series. 

But before that we briefly recall how notions of iteration arise in
the quantale setting, since these are needed for while rules in Hoare
logic. 

Since quantales are complete lattices, least and greatest fixpoints of
isotone functions exist. Moreover, due to their infinite
distributivity laws, functions such as $\lambda \alpha.\ x+ \alpha$,
$\lambda \alpha.\ x\cdot \alpha$ and $\lambda \alpha.\ \alpha\cdot x$
are continuous and the first one is even co-continuous in distributive
quantales. This means that in particular the least fixpoints built by
using combinations of these functions can be obtained by iteration
from $0$ to the first limit ordinal.

More specifically, the function $\varphi=\lambda\alpha.\ 1+x\cdot
\alpha$ is continuous, hence has the least fixpoint
$\mu\varphi=x^\ast=\sum_{i\in\mathbb{N}}\varphi^i(0)=\sum_{i\in\mathbb{N}}x^i$. This
notion of finite iteration is needed for deriving a while-rule for a
finite loop in a partial correctness setting.

More generally, the unfold and induction rules
\begin{alignat*}{4}
  1+x\cdot x^\ast &= x^\ast,&\qquad z+x\cdot y \le y&\Rightarrow 
  x^\ast\cdot z\le y,\\
  1+x^\ast\cdot x &= x^\ast,&\qquad z+y\cdot x \le y&\Rightarrow 
  z\cdot x^\ast\le y
\end{alignat*}
can be used for reasoning about the star. In a total correctness
setting, a notion of possibly infinite iteration is preferable, which
corresponds to the greatest fixpoint of $\varphi$. Infinite iteration
is also useful for futuristic monoids \refsec{sec:futuristic}, for
example, when reasoning about reactive systems, and Hoare rules for
these can be developed. However, because these follow a similar
pattern to finite iteration, we leave their full treatment as future
work.


Equipped with the star in the power series quantale we can now
follow~\cite{HMSW11} in setting up a propositional Hoare logic. The
development is slightly non-standard, in that there is no distinction
between assertions and programs at the level of algebra. It follows
the lines of a previous approach by Tarlecki~\cite{tarlecki}.
 
For a quantale $Q$ and elements $x,y,z\in Q$, we define validity of a
Hoare triple Tarlecki-style as
\begin{equation*}
  \vdash\{x\}y\{z\} \Leftrightarrow x\cdot y\le z.
\end{equation*}
In Tarlecki's original article, this encoding has been used for a
relational semantics where not only the program, but also its pre- and
postconditions are modelled as relations.  It is equally suitable for
trace or language based extensions of Hoare logic to concurrency, such
as the rely-guarantee method~\cite{Jon83}.

The proof of the following proposition is then straightforward and
generic for quantales.
\begin{proposition}[\cite{HMSW11}]\label{prop:phl}
  Let $Q$ be a unital quantale with unit $1$. The following
  rules of propositional Hoare logic are derivable, for all
  $w, w_1, w_2,x,x_1,x_2,y,y_1,y_2,z,z_1,z_2\in Q$.
  \begin{gather*}
\vdash\{x\}1\{x\}
\qquad
\frac{x_1\le x_2\quad\vdash \{x_2\}y\{z_2\}\quad z_2\le z_1}{\vdash\{x_1\}y\{z_1\}}
\\
    \frac{\vdash\{x\}y_1\{z\}\quad\vdash\{x\}y_2\{z\}}{\vdash\{x\}y_1+y_2\{z\}}
\qquad
\frac{\vdash\{w\}x_1\{z\}\quad\vdash\{z\}x_2\{y\}}{\vdash\{w\}x_1\cdot x_2\{y\}}
\\
\frac{\vdash\{x\}y\{x\}}{\vdash\{x\}y^\ast\{x\}}
  \end{gather*}
\end{proposition}
We can strengthen the choice and star rule  as follows.
\begin{equation*}
  \frac{\vdash\{x\cdot w_1\} y_1\{z\}\quad\vdash\{x\cdot w_2\}y_2\{z\}}{\vdash\{x\}w_1\cdot y_1+w_2\cdot y_2\{z\}}
\qquad 
\frac{\vdash\{x\cdot w_1\}y\{x\}}{\vdash\{x\}(w_1\cdot y)^\ast\cdot w_2\{x\cdot w_2\}}
\end{equation*}
The proof of the first one is essentially that of the choice rule. For
the second one suppose $x\cdot w_1\cdot y\le x$. Then $x\cdot
(w_1\cdot y)^\ast \le x$ by star induction and $x\cdot (w_1\cdot
y)^\ast\cdot w_2\le x\cdot w_2$ by isotonicity. If $w_1$ and $w_2$
are, in some sense, complemented, then this yields the standard
conditional rule and while rule of Hoare logic.

Instantiating Proposition~\ref{prop:phl} to power series quantales
automatically yields Hoare calculi for virtually all the examples
discussed in this article. The instantiation to the binary relations
quantale reproduces Tarlecki's original soundness result. Other
instances yield, in a generic way, Hoare logics over computationally
meaningful semantics based on finite words (traces in the sense of
concurrency theory), paths in graphs (sequences of events in
concurrency theory), paths in the sense of automata theory, or
pomsets. We also obtain generic propositional Hoare logics for
reasoning about interval and stream interval predicates in algebraic
variants of interval logics.


In addition, Proposition~\ref{prop:phl} covers commutative quantales,
where the Tarlecki-style encoding of the validity of Hoare triples
might make less sense. 

The rules covered by Proposition~\ref{prop:phl}, however, are entirely
sequential. For applications involving concurrency, such as the vector
stream interval functions in
Example~\ref{ex:vector-stream-interval-functions}, additional rules are
desirable.  In concurrent Kleene algebra, Owicki-Gries-style
concurrency rules and frame rules in the style of separation logic can
be derived. The same derivation, however, is ruled out in the quantale
context, because the concurrency rule obtained is equivalent to the
weak interchange law and the frame rule to one of the small
interchange laws, both of which have been refuted in
Proposition~\ref{prop:interchangeref}.

Instead we can use the interchange laws provided by
Lemma~\ref{P:quantaleprops}.
\begin{lemma}\label{lem:quantale-concrule}
  In quantale $Q$ the following concurrency rule is derivable,
  for all $x_1,x_2,$ $y_1,y_2,$ $z_1,z_2\in Q$.
\begin{equation*}
  \frac{\vdash \{x_1\}y_1\{z_1\}\quad \vdash \{x_2\}y_2\{z_2\}}{\vdash \{x_1\sqcap x_2\}y_1\sqcap y_2\{z_1\sqcap z_2\}}
\end{equation*}
\end{lemma}
\begin{proof}
Suppose $x_1\cdot y_1\le z_1$ and $x_2\cdot y_2\le
z_2$. Then 
\begin{equation*}
    (x_1\sqcap x_2)\cdot (y_1\sqcap y_2) 
  \le (x_1\cdot y_1)\sqcap (x_2 \cdot y_2)
  \le z_1\sqcap z_2
\end{equation*}
by weak interchange (Lemma~\ref{P:quantaleprops}) and the assumptions.
\end{proof}
\noindent Once more this rule is available automatically in all examples
discussed in this article.

As an alternative to conjunction-based notions of concurrency, it
might still be possible to derive concurrency and frame rules under
additional syntactic restrictions, for instance, those capturing the
synchronisation between sequential and concurrent compositions, or in
particular models. An investigation is left for future work.


\section{The Frame Rule in a Power Series Context}\label{sec:frame}

Section~\ref{sec:fpscomquantaleexamples} shows that the assertion
quantales which underlie separation logic---implementing the boolean
operations together with a notion of separation logic on predicates
over a resource monoid---can be modelled in the power series setting.
Predicate transformers, which yield another way of deriving Hoare
logics over assertion algebras, can be modelled in that setting as
well (Section~\ref{sec:transformers}).

In this section we sketch how a combination of these results allows us
to derive the frame rule of separation logic by equational
reasoning. Convolution plays a central part in the proof.  Previously,
algebraic proofs of the frame rule have been given in a state
transformer context~\cite{COY07} as well as in the context of
concurrent Kleene algebra~\cite{HMSW11}.

It is well known that in the predicate transformer setting, validity
of Hoare triples can be encoded as
\begin{equation*}
  \vdash \{p\}R\{q\}\Leftrightarrow p\le \hat{f}_R\, q,
\end{equation*}
which is essentially an adjunction, using the notation of Section~\ref{sec:transformers},
but writing $p,q,\dots$ for predicates, which are elements of the
assertion quantale of separation logic. It is also well known that the
rules of Hoare logic can be derived in this setting, assuming that
predicate transformers are isotone. A result of separation logic
states that the frame rule can be derived whenever the predicate
transformer $f$ under consideration is \emph{local}, that is, it
satisfies
\begin{equation*}
  f\ast\mathit{id}\le f.
\end{equation*}
Intuitively, locality means that the effect of a transformer can
always be localised on part of the state. For a detailed discussion
see~\cite{COY07}.

Before deriving the frame rule we use properties of power series and
convolution to prove a point-wise analogue of locality which simplfies
the proof.
\begin{lemma}\label{lem:local-prop}
  $f$ is local if and only if $(f\, p) \ast q \le f\, (p\ast q)$.
\end{lemma}
\begin{proof}
  Let $(f\, p)\ast q\le  f\, (p\ast q)$. Then 
$(f\, p)\ast (\mathit{id}\, q) = (f\, p) \ast q \le f\, (p\ast q)$
 and therefore 
\begin{equation*}
    (f\ast \mathit{id})\, r 
  = \sum_{r=p\ast q} (f\, p)\ast (\mathit{id}\, q)
  \le \sum_{r=p\ast q}  f\, (p\ast q) = f\, r.
\end{equation*}

Let $f$ be local.  Then 
\begin{equation*}
(f\ast \mathit{id})\, r = \sum_{r=p\ast q} (f\, p) \ast q \le f\, r = f\, (p\ast q),
\end{equation*}
 whence $(f\, p)\ast q\le f\, (p\ast q)$.
\end{proof}

\begin{lemma}
  Let $\hat{f}_R$ be a local predicate transformer associated to program
  $R$. Then the following frame rule holds.
  \begin{equation*}
    \frac{\vdash\{p\}R\{q\}}{\vdash\{p\ast r\}R\{q\ast r\}}
  \end{equation*}
\end{lemma}
\begin{proof} 
  Let $\vdash\{p\}R\{q\}$, that is, $p\le \hat{f}_R\, q$. Then $p \ast
  r \le (\hat{f}_R\, q) \ast r \le \hat{f}_R\, (q \ast r)$ by
  Lemma~\ref{lem:local-prop} and therefore $\vdash\{p\ast r\}R\{q\ast
  r\}$.
\end{proof}
A deeper investigation of Hoare logics, inference rules for separation
logic, and extensions to concurrency in this setting is left for
future work.


\section{Conclusion}\label{sec:conclusion}

The aim of this article is to demonstrate that convolution is a
versatile and interesting construction in mathematics and computer
science.  Used in the context of power series and integrated into
lifting results, it yields a powerful tool for setting up various
mathematical structures and computational models and calculi endowed
with generic algebraic properties.

Beyond the language models known from formal language theory, these
include assertion quantales of separation logic (which can be lifted
from an underlying resource monoid), assertion quantales of interval
logics (which can be lifted from an underlying semigroup of intervals)
and stream interval functions (which have applications in the analysis
of dynamic and real-time systems). For all these examples, the power
series approach provides a simple new approach. For the latter two,
new kinds of concurrency operations are provided.

In addition, the modelling framework based on power series has been
combined with a verification approach by deriving, in generic fashion,
propositional Hoare logics for virtually all the examples
considered. In particular, state, predicate or resource transformers,
which can be used for constructing these logics, arise as instances of
power series.

This article focused mainly on the proof of concept of the relevance
of convolution. Many of the modelling examples and verification
approaches featured require further investigation. This includes in
particular the derivation of more comprehensive sets of Hoare-style
inference rules for concurrency verification, separation logic and
interval temporal logics, and more detailed case studies with
separation, inverval and stream interval algebras, and with concurrent
systems with infinite behaviours.

For all these case studies, the formalisation of the power series
approach and the implementation of modelling tools plays and important
role.  In fact, the basic lifting lemma and a detailed predicate
transformer approach based on power series have already been
formalised within the Isabelle/HOL proof assistant~\cite{NPW02}.  The
development of a power series based verification tool for separation
logic, and even concurrent separation logic, will be the next step in
the tool chain.

\bibliographystyle{plain}
\bibliography{interval}

\end{document}

%% file: chop.pspdftex
\begin{picture}(0,0)%
\includegraphics{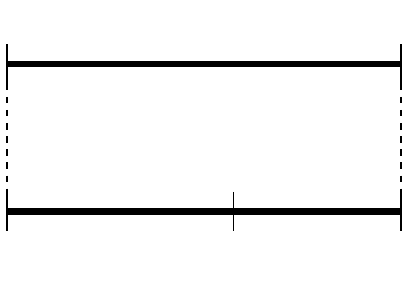}%
\end{picture}%
\setlength{\unitlength}{4144sp}%
\begingroup\makeatletter\ifx\SetFigFont\undefined%
\gdef\SetFigFont#1#2#3#4#5{%
  \reset@font\fontsize{#1}{#2pt}%
  \fontfamily{#3}\fontseries{#4}\fontshape{#5}%
  \selectfont}%
\fi\endgroup%
\begin{picture}(1866,1291)(1318,-1385)
\put(1351,-241){\makebox(0,0)[b]{\smash{{\SetFigFont{12}{14.4}{\rmdefault}{\mddefault}{\updefault}{\color[rgb]{0,0,0}$a$}%
}}}}
\put(3151,-241){\makebox(0,0)[b]{\smash{{\SetFigFont{12}{14.4}{\rmdefault}{\mddefault}{\updefault}{\color[rgb]{0,0,0}$c$}%
}}}}
\put(2251,-286){\makebox(0,0)[b]{\smash{{\SetFigFont{12}{14.4}{\rmdefault}{\mddefault}{\updefault}{\color[rgb]{0,0,0}$f \cdot g$}%
}}}}
\put(2791,-1276){\makebox(0,0)[b]{\smash{{\SetFigFont{12}{14.4}{\rmdefault}{\mddefault}{\updefault}{\color[rgb]{0,0,0}$g$}%
}}}}
\put(2386,-1321){\makebox(0,0)[b]{\smash{{\SetFigFont{12}{14.4}{\rmdefault}{\mddefault}{\updefault}{\color[rgb]{0,0,0}$b$}%
}}}}
\put(1891,-1276){\makebox(0,0)[b]{\smash{{\SetFigFont{12}{14.4}{\rmdefault}{\mddefault}{\updefault}{\color[rgb]{0,0,0}$f$}%
}}}}
\put(1351,-1276){\makebox(0,0)[b]{\smash{{\SetFigFont{12}{14.4}{\rmdefault}{\mddefault}{\updefault}{\color[rgb]{0,0,0}$a$}%
}}}}
\put(3151,-1276){\makebox(0,0)[b]{\smash{{\SetFigFont{12}{14.4}{\rmdefault}{\mddefault}{\updefault}{\color[rgb]{0,0,0}$c$}%
}}}}
\end{picture}%

%% file: tcube.pspdftex
\begin{picture}(0,0)%
\includegraphics{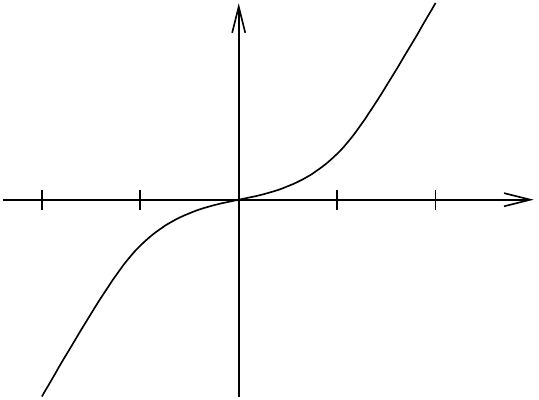}%
\end{picture}%
\setlength{\unitlength}{4144sp}%
\begingroup\makeatletter\ifx\SetFigFont\undefined%
\gdef\SetFigFont#1#2#3#4#5{%
  \reset@font\fontsize{#1}{#2pt}%
  \fontfamily{#3}\fontseries{#4}\fontshape{#5}%
  \selectfont}%
\fi\endgroup%
\begin{picture}(2454,1824)(3409,-3223)
\put(4051,-2221){\makebox(0,0)[b]{\smash{{\SetFigFont{8}{9.6}{\rmdefault}{\mddefault}{\updefault}{\color[rgb]{0,0,0}$-5$}%
}}}}
\put(5401,-2491){\makebox(0,0)[b]{\smash{{\SetFigFont{8}{9.6}{\rmdefault}{\mddefault}{\updefault}{\color[rgb]{0,0,0}$10$}%
}}}}
\put(3601,-2221){\makebox(0,0)[b]{\smash{{\SetFigFont{8}{9.6}{\rmdefault}{\mddefault}{\updefault}{\color[rgb]{0,0,0}$-10$}%
}}}}
\put(4951,-2491){\makebox(0,0)[b]{\smash{{\SetFigFont{8}{9.6}{\rmdefault}{\mddefault}{\updefault}{\color[rgb]{0,0,0}$5$}%
}}}}
\put(5806,-2221){\makebox(0,0)[b]{\smash{{\SetFigFont{10}{12.0}{\rmdefault}{\mddefault}{\updefault}{\color[rgb]{0,0,0}$t$}%
}}}}
\put(5401,-1636){\makebox(0,0)[b]{\smash{{\SetFigFont{10}{12.0}{\rmdefault}{\mddefault}{\updefault}{\color[rgb]{0,0,0}$f$}%
}}}}
\end{picture}%

%% file: mult-2dim.pspdftex
\begin{picture}(0,0)%
\includegraphics{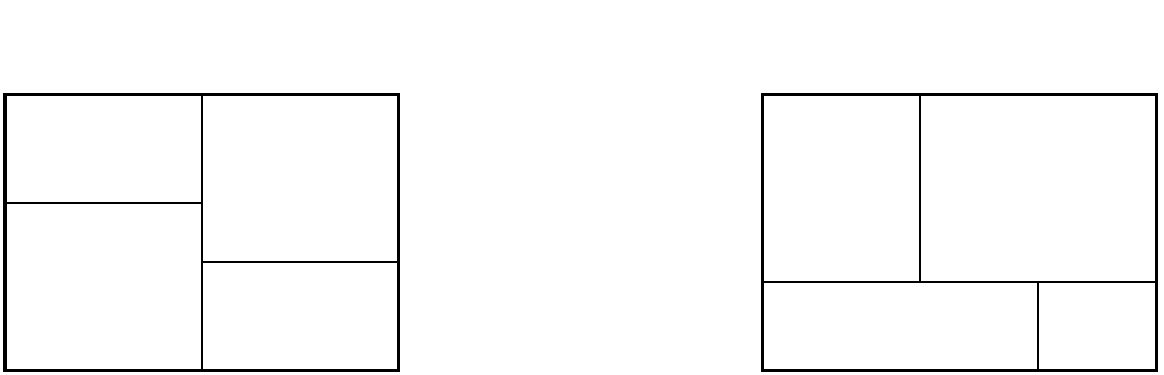}%
\end{picture}%
\setlength{\unitlength}{4144sp}%
\begingroup\makeatletter\ifx\SetFigFont\undefined%
\gdef\SetFigFont#1#2#3#4#5{%
  \reset@font\fontsize{#1}{#2pt}%
  \fontfamily{#3}\fontseries{#4}\fontshape{#5}%
  \selectfont}%
\fi\endgroup%
\begin{picture}(5309,1714)(-21,-533)
\put(1351,254){\makebox(0,0)[b]{\smash{{\SetFigFont{12}{14.4}{\rmdefault}{\mddefault}{\updefault}{\color[rgb]{0,0,0}$H$}%
}}}}
\put(4321,974){\makebox(0,0)[b]{\smash{{\SetFigFont{14}{16.8}{\rmdefault}{\mddefault}{\updefault}{\color[rgb]{0,0,0}$(F \cdot H) \ast (G \cdot K)$}%
}}}}
\put(451,434){\makebox(0,0)[b]{\smash{{\SetFigFont{12}{14.4}{\rmdefault}{\mddefault}{\updefault}{\color[rgb]{0,0,0}$F$}%
}}}}
\put(1351,-331){\makebox(0,0)[b]{\smash{{\SetFigFont{12}{14.4}{\rmdefault}{\mddefault}{\updefault}{\color[rgb]{0,0,0}$K$}%
}}}}
\put(4996,-376){\makebox(0,0)[b]{\smash{{\SetFigFont{12}{14.4}{\rmdefault}{\mddefault}{\updefault}{\color[rgb]{0,0,0}$K$}%
}}}}
\put(3871,254){\makebox(0,0)[b]{\smash{{\SetFigFont{12}{14.4}{\rmdefault}{\mddefault}{\updefault}{\color[rgb]{0,0,0}$F$}%
}}}}
\put(4096,-376){\makebox(0,0)[b]{\smash{{\SetFigFont{12}{14.4}{\rmdefault}{\mddefault}{\updefault}{\color[rgb]{0,0,0}$G$}%
}}}}
\put(4726,254){\makebox(0,0)[b]{\smash{{\SetFigFont{12}{14.4}{\rmdefault}{\mddefault}{\updefault}{\color[rgb]{0,0,0}$H$}%
}}}}
\put(451,-151){\makebox(0,0)[b]{\smash{{\SetFigFont{12}{14.4}{\rmdefault}{\mddefault}{\updefault}{\color[rgb]{0,0,0}$G$}%
}}}}
\put(901,974){\makebox(0,0)[b]{\smash{{\SetFigFont{14}{16.8}{\rmdefault}{\mddefault}{\updefault}{\color[rgb]{0,0,0}$(F \ast G) \cdot (H \ast K)$}%
}}}}
\end{picture}%

%% file: no-weak-inter-3.pspdftex
\begin{picture}(0,0)%
\includegraphics{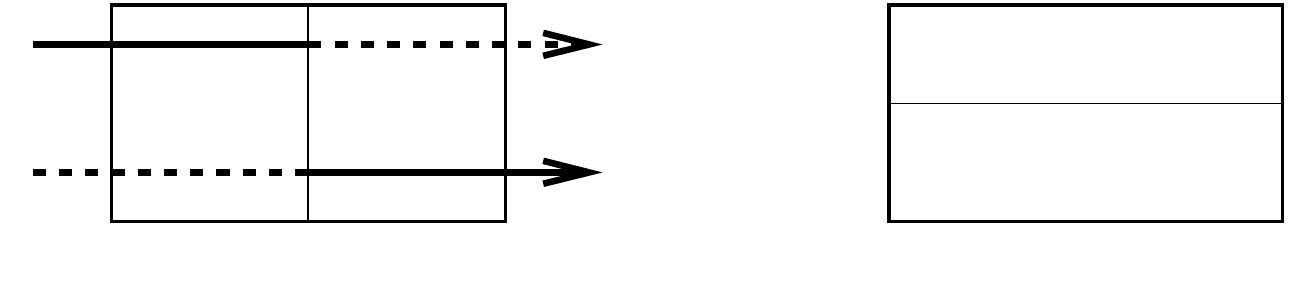}%
\end{picture}%
\setlength{\unitlength}{4144sp}%
\begingroup\makeatletter\ifx\SetFigFont\undefined%
\gdef\SetFigFont#1#2#3#4#5{%
  \reset@font\fontsize{#1}{#2pt}%
  \fontfamily{#3}\fontseries{#4}\fontshape{#5}%
  \selectfont}%
\fi\endgroup%
\begin{picture}(5887,1301)(-509,-530)
\put(-494,-61){\makebox(0,0)[b]{\smash{{\SetFigFont{12}{14.4}{\rmdefault}{\mddefault}{\updefault}{\color[rgb]{0,0,0}$f_2$}%
}}}}
\put(-494,524){\makebox(0,0)[b]{\smash{{\SetFigFont{12}{14.4}{\rmdefault}{\mddefault}{\updefault}{\color[rgb]{0,0,0}$f_1$}%
}}}}
\put(451,209){\makebox(0,0)[b]{\smash{{\SetFigFont{12}{14.4}{\rmdefault}{\mddefault}{\updefault}{\color[rgb]{0,0,0}$F$}%
}}}}
\put(1351,209){\makebox(0,0)[b]{\smash{{\SetFigFont{12}{14.4}{\rmdefault}{\mddefault}{\updefault}{\color[rgb]{0,0,0}$G$}%
}}}}
\put(4456,434){\makebox(0,0)[b]{\smash{{\SetFigFont{12}{14.4}{\rmdefault}{\mddefault}{\updefault}{\color[rgb]{0,0,0}$F$}%
}}}}
\put(4456,-61){\makebox(0,0)[b]{\smash{{\SetFigFont{12}{14.4}{\rmdefault}{\mddefault}{\updefault}{\color[rgb]{0,0,0}$G$}%
}}}}
\put(  1,-466){\makebox(0,0)[b]{\smash{{\SetFigFont{12}{14.4}{\rmdefault}{\mddefault}{\updefault}{\color[rgb]{0,0,0}$-10$}%
}}}}
\put(1801,-466){\makebox(0,0)[b]{\smash{{\SetFigFont{12}{14.4}{\rmdefault}{\mddefault}{\updefault}{\color[rgb]{0,0,0}$10$}%
}}}}
\put(901,-466){\makebox(0,0)[b]{\smash{{\SetFigFont{12}{14.4}{\rmdefault}{\mddefault}{\updefault}{\color[rgb]{0,0,0}$0$}%
}}}}
\put(3556,-466){\makebox(0,0)[b]{\smash{{\SetFigFont{12}{14.4}{\rmdefault}{\mddefault}{\updefault}{\color[rgb]{0,0,0}$-10$}%
}}}}
\put(5356,-466){\makebox(0,0)[b]{\smash{{\SetFigFont{12}{14.4}{\rmdefault}{\mddefault}{\updefault}{\color[rgb]{0,0,0}$10$}%
}}}}
\end{picture}%

%% file: no-weak-inter-4.pspdftex
\begin{picture}(0,0)%
\includegraphics{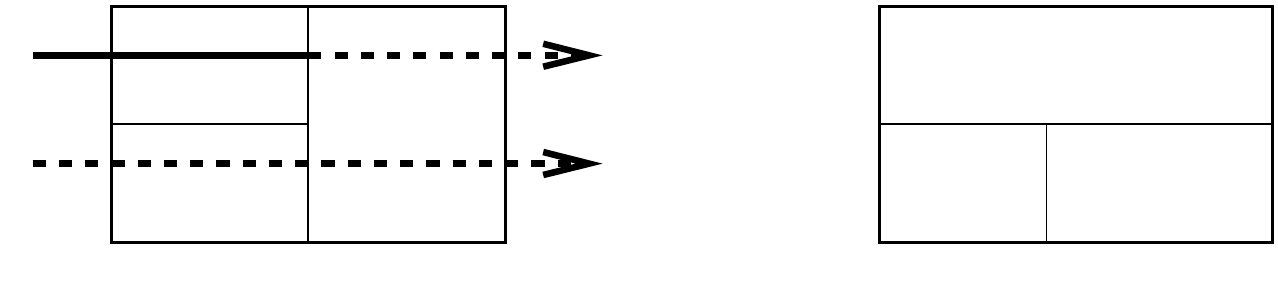}%
\end{picture}%
\setlength{\unitlength}{4144sp}%
\begingroup\makeatletter\ifx\SetFigFont\undefined%
\gdef\SetFigFont#1#2#3#4#5{%
  \reset@font\fontsize{#1}{#2pt}%
  \fontfamily{#3}\fontseries{#4}\fontshape{#5}%
  \selectfont}%
\fi\endgroup%
\begin{picture}(5842,1391)(-509,-620)
\put(1351,254){\makebox(0,0)[b]{\smash{{\SetFigFont{12}{14.4}{\rmdefault}{\mddefault}{\updefault}{\color[rgb]{0,0,0}$H$}%
}}}}
\put(-494,-61){\makebox(0,0)[b]{\smash{{\SetFigFont{12}{14.4}{\rmdefault}{\mddefault}{\updefault}{\color[rgb]{0,0,0}$f_2$}%
}}}}
\put(451,299){\makebox(0,0)[b]{\smash{{\SetFigFont{12}{14.4}{\rmdefault}{\mddefault}{\updefault}{\color[rgb]{0,0,0}$F$}%
}}}}
\put(-494,524){\makebox(0,0)[b]{\smash{{\SetFigFont{12}{14.4}{\rmdefault}{\mddefault}{\updefault}{\color[rgb]{0,0,0}$f_1$}%
}}}}
\put(451,-196){\makebox(0,0)[b]{\smash{{\SetFigFont{12}{14.4}{\rmdefault}{\mddefault}{\updefault}{\color[rgb]{0,0,0}$G$}%
}}}}
\put(  1,-556){\makebox(0,0)[b]{\smash{{\SetFigFont{12}{14.4}{\rmdefault}{\mddefault}{\updefault}{\color[rgb]{0,0,0}$-10$}%
}}}}
\put(1801,-556){\makebox(0,0)[b]{\smash{{\SetFigFont{12}{14.4}{\rmdefault}{\mddefault}{\updefault}{\color[rgb]{0,0,0}$10$}%
}}}}
\put(901,-556){\makebox(0,0)[b]{\smash{{\SetFigFont{12}{14.4}{\rmdefault}{\mddefault}{\updefault}{\color[rgb]{0,0,0}$0$}%
}}}}
\put(3511,-556){\makebox(0,0)[b]{\smash{{\SetFigFont{12}{14.4}{\rmdefault}{\mddefault}{\updefault}{\color[rgb]{0,0,0}$-10$}%
}}}}
\put(5311,-556){\makebox(0,0)[b]{\smash{{\SetFigFont{12}{14.4}{\rmdefault}{\mddefault}{\updefault}{\color[rgb]{0,0,0}$10$}%
}}}}
\put(4411,434){\makebox(0,0)[b]{\smash{{\SetFigFont{12}{14.4}{\rmdefault}{\mddefault}{\updefault}{\color[rgb]{0,0,0}$F$}%
}}}}
\put(4771,-106){\makebox(0,0)[b]{\smash{{\SetFigFont{12}{14.4}{\rmdefault}{\mddefault}{\updefault}{\color[rgb]{0,0,0}$H$}%
}}}}
\put(3871,-106){\makebox(0,0)[b]{\smash{{\SetFigFont{12}{14.4}{\rmdefault}{\mddefault}{\updefault}{\color[rgb]{0,0,0}$G$}%
}}}}
\end{picture}%

%% file: no-weak-inter-2.pspdftex
\begin{picture}(0,0)%
\includegraphics{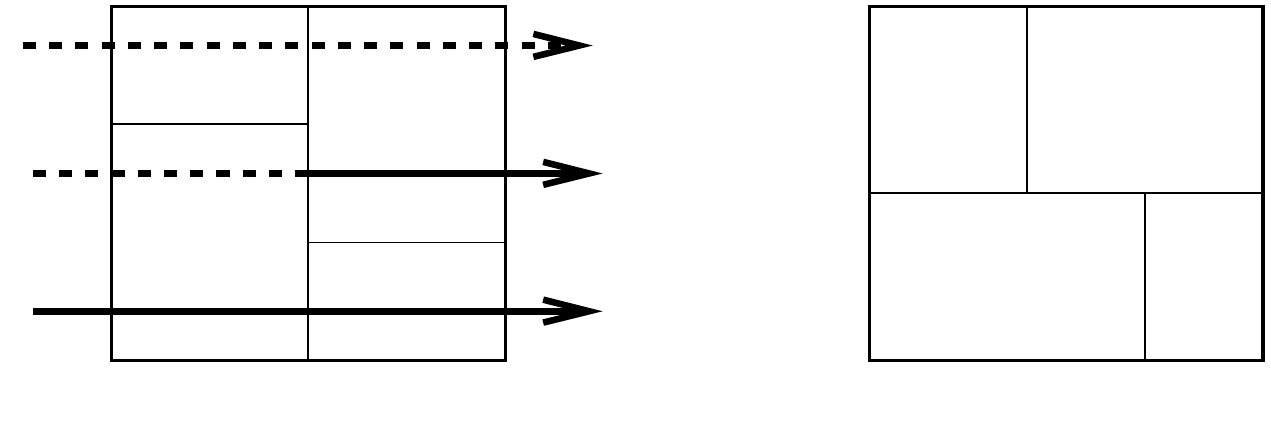}%
\end{picture}%
\setlength{\unitlength}{4144sp}%
\begingroup\makeatletter\ifx\SetFigFont\undefined%
\gdef\SetFigFont#1#2#3#4#5{%
  \reset@font\fontsize{#1}{#2pt}%
  \fontfamily{#3}\fontseries{#4}\fontshape{#5}%
  \selectfont}%
\fi\endgroup%
\begin{picture}(5797,1931)(-509,-1160)
\put(451,-421){\makebox(0,0)[b]{\smash{{\SetFigFont{12}{14.4}{\rmdefault}{\mddefault}{\updefault}{\color[rgb]{0,0,0}$G$}%
}}}}
\put(1351,-601){\makebox(0,0)[b]{\smash{{\SetFigFont{12}{14.4}{\rmdefault}{\mddefault}{\updefault}{\color[rgb]{0,0,0}$K$}%
}}}}
\put(1351,254){\makebox(0,0)[b]{\smash{{\SetFigFont{12}{14.4}{\rmdefault}{\mddefault}{\updefault}{\color[rgb]{0,0,0}$H$}%
}}}}
\put(-494,-61){\makebox(0,0)[b]{\smash{{\SetFigFont{12}{14.4}{\rmdefault}{\mddefault}{\updefault}{\color[rgb]{0,0,0}$f_2$}%
}}}}
\put(-494,-691){\makebox(0,0)[b]{\smash{{\SetFigFont{12}{14.4}{\rmdefault}{\mddefault}{\updefault}{\color[rgb]{0,0,0}$f_3$}%
}}}}
\put(4771,299){\makebox(0,0)[b]{\smash{{\SetFigFont{12}{14.4}{\rmdefault}{\mddefault}{\updefault}{\color[rgb]{0,0,0}$H$}%
}}}}
\put(4051,-601){\makebox(0,0)[b]{\smash{{\SetFigFont{12}{14.4}{\rmdefault}{\mddefault}{\updefault}{\color[rgb]{0,0,0}$G$}%
}}}}
\put(3826,299){\makebox(0,0)[b]{\smash{{\SetFigFont{12}{14.4}{\rmdefault}{\mddefault}{\updefault}{\color[rgb]{0,0,0}$F$}%
}}}}
\put(4996,-601){\makebox(0,0)[b]{\smash{{\SetFigFont{12}{14.4}{\rmdefault}{\mddefault}{\updefault}{\color[rgb]{0,0,0}$K$}%
}}}}
\put(451,299){\makebox(0,0)[b]{\smash{{\SetFigFont{12}{14.4}{\rmdefault}{\mddefault}{\updefault}{\color[rgb]{0,0,0}$F$}%
}}}}
\put(-494,524){\makebox(0,0)[b]{\smash{{\SetFigFont{12}{14.4}{\rmdefault}{\mddefault}{\updefault}{\color[rgb]{0,0,0}$f_1$}%
}}}}
\put(  1,-1096){\makebox(0,0)[b]{\smash{{\SetFigFont{12}{14.4}{\rmdefault}{\mddefault}{\updefault}{\color[rgb]{0,0,0}$-10$}%
}}}}
\put(901,-1096){\makebox(0,0)[b]{\smash{{\SetFigFont{12}{14.4}{\rmdefault}{\mddefault}{\updefault}{\color[rgb]{0,0,0}$0$}%
}}}}
\put(1801,-1096){\makebox(0,0)[b]{\smash{{\SetFigFont{12}{14.4}{\rmdefault}{\mddefault}{\updefault}{\color[rgb]{0,0,0}$10$}%
}}}}
\put(3466,-1096){\makebox(0,0)[b]{\smash{{\SetFigFont{12}{14.4}{\rmdefault}{\mddefault}{\updefault}{\color[rgb]{0,0,0}$-10$}%
}}}}
\put(5266,-1096){\makebox(0,0)[b]{\smash{{\SetFigFont{12}{14.4}{\rmdefault}{\mddefault}{\updefault}{\color[rgb]{0,0,0}$10$}%
}}}}
\end{picture}%

%% file: interval4.bbl
\begin{thebibliography}{10}

\bibitem{AbramskyJung}
S.~Abramsky and A.~Jung.
\newblock Domain theory.
\newblock In S.~Abramsky, D.~M. Gabbay, and T.~S.~E. Maibaum, editors, {\em
  Handbook of Logic in Computer Science}, volume III. Oxford University Press,
  1994.

\bibitem{ArmstrongSW13}
A.~Armstrong, G.~Struth, and T.~Weber.
\newblock Kleene algebra.
\newblock {\em Archive of Formal Proofs}, 2013.

\bibitem{ArmstrongSW-JLAMP}
A.~Armstrong, G.~Struth, and T.~Weber.
\newblock Programming and automating mathematics in the {T}arski-{K}leene
  hierarchy.
\newblock {\em Journal of Logical and Algebraic Methods in Programming},
  83(2):87--102, 2014.

\bibitem{BvW99-book}
R.-J. Back and J.~von Wright.
\newblock {\em Refinement calculus - a systematic introduction}.
\newblock Springer, 1999.

\bibitem{BerstelReutenauer}
J.~Berstel and C.~Reutenauer.
\newblock {\em Les s\'eries rationnelles et leurs langagues}.
\newblock Masson, 1984.

\bibitem{Brink}
C.~Brink.
\newblock Power structures.
\newblock {\em Algebra Universalis}, 30:177--216, 1993.

\bibitem{COY07}
C.~Calcagno, P.~W. O'Hearn, and H.~Yang.
\newblock Local action and abstract separation logic.
\newblock In {\em LICS}, pages 366--378. IEEE Computer Society, 2007.

\bibitem{Conway71}
J.~H. Conway.
\newblock {\em Regular Algebra and Finite Machines}.
\newblock Chapman and Hall, 1971.

\bibitem{Day}
B.~Day.
\newblock On closed categories of functors.
\newblock In {\em Reports of the Midwest Category Seminar IV}, volume 137 of
  {\em Lecture Notes in Mathematics}, pages 1--38. Springer, 1970.

\bibitem{D-YBGPY13}
T.~Dinsdale-Young, L.~Birkedal, P.~Gardner, M.~J. Parkinson, and H.~Yang.
\newblock Views: compositional reasoning for concurrent programs.
\newblock In R.~Giacobazzi and R.~Cousot, editors, {\em POPL}, pages 287--300.
  ACM, 2013.

\bibitem{DHD14}
B.~Dongol, I.~J. Hayes, and J.~Derrick.
\newblock Deriving real-time action systems with multiple time bands using
  algebraic reasoning.
\newblock {\em Sci. Comput. Program.}, 85:137--165, 2014.

\bibitem{Handbook}
M.~Droste, W.~Kuich, and H.~Vogler, editors.
\newblock {\em Handbook of Weighted Automata}.
\newblock Springer, 2009.

\bibitem{Eilenberg}
S.~Eilenberg.
\newblock {\em Automata, Languages and Machines}, volume~A.
\newblock Academic Press, 1974.

\bibitem{Gischer}
J.~L. Gischer.
\newblock The equational theory of pomsets.
\newblock {\em Theoretical Computer Science}, 61:199--224, 1988.

\bibitem{Goldblatt}
R.~Goldblatt.
\newblock Varieties of complex algebras.
\newblock {\em Annals of Pure and Applied Logic}, 44:173--242, 1989.

\bibitem{GondranMinoux}
M.~Gondran and M.~Minoux.
\newblock {\em Graphs, Dioids and Semirings}.
\newblock Springer, 2008.

\bibitem{Grabowski}
J.~Grabowski.
\newblock On partial languages.
\newblock {\em Fundamentae Informaticae}, 4:427--498, 1981.

\bibitem{HBDJ13}
I.~J. Hayes, A.~Burns, B.~Dongol, and C.~B. Jones.
\newblock Comparing degrees of non-determinism in expression evaluation.
\newblock {\em Comput. J.}, 56(6):741--755, 2013.

\bibitem{Locality}
C.~A.~R. Hoare, A.~Hussain, B.~M{\"{o}}ller, P.~W. O'Hearn,
  R.~Lerchedahl~Petersen, and G.~Struth.
\newblock On locality and the exchange law for concurrent processes.
\newblock In J.-P. Katoen and B.~K{\"{o}}nig, editors, {\em {CONCUR} 2011},
  volume 6901 of {\em LNCS}, pages 250--264. Springer, 2011.

\bibitem{HMSW11}
T.~Hoare, B.~M{\"o}ller, G.~Struth, and I.~Wehrman.
\newblock Concurrent {Kleene} algebra and its foundations.
\newblock {\em J. Log. Algebr. Program.}, 80(6):266--296, 2011.

\bibitem{HM09}
P.~H{\"o}fner and B.~M{\"o}ller.
\newblock An algebra of hybrid systems.
\newblock {\em J. Log. Algebr. Program.}, 78(2):74--97, 2009.

\bibitem{Jon83}
C.~B. Jones.
\newblock Tentative steps toward a development method for interfering programs.
\newblock {\em ACM Transactions on Programming Languages and Systems},
  5(4):596--619, 1983.

\bibitem{Kelly}
G.~M. Kelly.
\newblock Basic concepts of enriched category theory.
\newblock {\em LMS Lecture Notes Series}, 64, 1982.

\bibitem{Kozen91}
D.~Kozen.
\newblock A completeness theorem for {K}leene algebras and the algebra of
  regular events.
\newblock In {\em LICS}, pages 214--225. IEEE Comp. Soc., 1991.

\bibitem{Mos00}
B.~C. Moszkowski.
\newblock A complete axiomatization of interval temporal logic with infinite
  time.
\newblock In {\em 15th Annual {IEEE} Symposium on Logic in Computer Science,
  Santa Barbara, California, USA, June 26-29, 2000}, pages 241--252. {IEEE}
  Computer Society, 2000.

\bibitem{NPW02}
T.~Nipkow, L.~C. Paulson, and M.~Wenzel.
\newblock {\em Isabelle/HOL --- A Proof Assistant for Higher-Order Logic},
  volume 2283 of {\em LNCS}.
\newblock Springer, 2002.

\bibitem{OHearnP99}
P.~W. O'Hearn and D.~J. Pym.
\newblock The logic of bunched implications.
\newblock {\em Bulletin of Symbolic Logic}, 5(2):215--244, 1999.

\bibitem{OG76}
S.~S. Owicki and D.~Gries.
\newblock Verifying properties of parallel programs: An axiomatic approach.
\newblock {\em Commun. ACM}, 19(5):279--285, 1976.

\bibitem{tarlecki}
A.~Tarlecki.
\newblock A language of specified programs.
\newblock {\em Science of Computer Programming}, 5:59--81, 1985.

\bibitem{ZH04}
C.~Zhou and M.~R. Hansen.
\newblock {\em Duration Calculus: A Formal Approach to Real-Time Systems}.
\newblock EATCS: Monographs in Theoretical Computer Science. Springer, 2004.

\end{thebibliography}
